\documentclass[a4paper,american,cleveref,autoref,thm-restate]{lipics-v2019}
\usepackage[utf8]{inputenc} 
\usepackage{xspace} 
\usepackage{latexsym} 
\usepackage{amssymb} 
\usepackage{amsmath} 
\usepackage{enumerate}
\usepackage{tcolorbox} 
\usepackage{tikz} 
\usetikzlibrary{snakes,shapes,arrows,automata}  
\usepackage{color} 

\newcounter{showLineNumbersLeft} 
\newcounter{showIndexEntries} 
\newcounter{showLabels} 

\usepackage{prettyref} 
\usepackage{comment}
\nolinenumbers
\usepackage[utf8]{inputenc} 
\usepackage{xspace} 
\usepackage{latexsym} 
\usepackage{amssymb} 
\usepackage{amsmath} 
\usepackage{ifthen} 
\usepackage{tikz} 
\usepackage{todonotes} 
\usepackage{color} 

\usepackage{subcaption}
\usepackage{prettyref} 

\newrefformat{thm}{Satz~\ref{#1}} 
\newrefformat{lem}{Lemma~\ref{#1}} 
\newrefformat{dfn}{Definition~\ref{#1}} 
\newrefformat{cor}{Korollar~\ref{#1}}  
\newrefformat{prop}{Proposition~\ref{#1}} 
\newrefformat{sec}{Abschnitt~\ref{#1}} 
\newrefformat{chp}{Kapitel~\ref{#1}} 
\newrefformat{bsp}{Beispiel~\ref{#1}} 
\newrefformat{rem}{Bemerkung~\ref{#1}} 
\newrefformat{fig}{Abbildung~\ref{#1}} 
\newrefformat{eq}{Gleichung~\eqref{#1}} 
\newrefformat{a}{Aufgabe~\ref{#1}} 
\newrefformat{ex}{Beispiel~\ref{#1}} 

\DeclareFontShape{OT1}{cmtt}{bx}{n} 
{<5><6><7><8>cmbtt8%
 <9>cmbtt9%
 <10><10.95>cmbtt10%
 <12><14.4><17.28><20.74><24.88>cmbtt10%
} 

\allowdisplaybreaks

\colorlet{DMnormalbackcolor}{gray!25} 
\colorlet{DMlightbackcolor}{gray!10} 
\colorlet{DMmediumbackcolor}{gray!20} 
\colorlet{DMdarkbackcolor}{gray!30} 
\colorlet{DMmediumforecolor}{black!57} 
\colorlet{DMlightforecolor}{black!46} 

\newlength{\ksize} 
\ifthenelse{\equal{\value{showLabels}}{1}}{%
 \usepackage[notref,notcite]{showkeys}%
}%
{} 

\ifthenelse{\equal{\value{showIndexEntries}}{1}}{%
 \usepackage{showidx}%
}%
{} 

\ifthenelse{\equal{\value{showLineNumbersLeft}}{1}}{%
 \usepackage[left]{lineno}%
 \setlength\linenumbersep{6pt}%
 \addtolength\marginparsep{\linenumbersep}%
 \addtolength\marginparsep{6pt}%
}%
{} 

\newcommand{\sset}[2]{\left\{\, \mathinner{#1}\vphantom{#2}\; \left|\; \vphantom{#1}\mathinner{#2} \right.\,\right\}}

\newcommand\slex{\mathrel{\leq_\mathrm{slex}}}  

\newcommand\lds{,\ldots ,}  
  
\newcommand{\sse}{\subseteq} 
 
\newcommand{\es}{\emptyset} 
\newcommand{\sm}{\setminus} 
\newcommand{\os}[1]{\{#1\}} 
\newcommand{\wh}[1]{\widehat{#1}}

\newcommand\sat{\mathop\mathrm{Sat}}

\newcommand{\btorp}{$b$-torsion property\xspace} 
\newcommand{\btptorp}{$(b,t,p)$-torsion property\xspace} 
\newcommand{\cfree}{context-free\xspace} 
\newcommand{\Cfree}{Context-free\xspace}
\renewcommand{\hom}{homo\-mor\-phism\xspace} 
 
\newcommand{\subst}{substitution\xspace} 
\newcommand{\morph}{mor\-phism\xspace}

\newcommand\WP{\mathop\mathrm{WP}}

\newcommand{\G}{\mathbb{G}} 
\newcommand{\N}{\mathbb{N}} 
\newcommand{\Z}{\mathbb{Z}} 
\newcommand{\V}{\mathbb{V}} 
 
\newcommand{\R}{\mathbb{R}}

\newcommand{\E}{\mathbb{E}}

\newcommand{\NP}{\mathbf{NP}}

\newcommand{\Oh}{\mathcal{O}} 
 
\newcommand{\cB}{\mathcal{B}} 
\newcommand{\cC}{\mathcal{C}} 
\newcommand{\cF}{\mathcal{F}} 
\newcommand{\cG}{\mathcal{G}}

\newcommand{\Sig}{\Sigma} 
 
\newcommand{\Gam}{\Gamma} 
 
\newcommand{\OO}{\Omega}

\newcommand{\eps}{\varepsilon} 
\newcommand{\alp}{\alpha} 
\newcommand{\bet}{\beta} 
\newcommand{\gam}{\gamma}

\newcommand{\ov}[1]{\overleftarrow{#1}} 
\newcommand{\oi}[1]{{#1}^{-1}}

\newcommand{\id}{\mathrm{id}}

\newcommand{\abs}[1]{\left|\mathinner{#1}\right|} 
\newcommand{\absbin}[1]{|\mathinner{#1}|_{\text{bin}}}

\newcommand{\FO}{\mathrm{FO}}

\newcommand\rf{\mathrm{rf}} 

\def\math#1{\ifmmode #1\else \mbox{$#1$} \fi}

\newcommand{\IFF}{if and only if\xspace} 

\newcommand{\set}[2]{\left\{#1\, \middle| \,#2 \right\}} 

\makeatletter 
\newcommand*{\TDFTA}{\@ifnextchar{.}{$\downarrow$FTA}{$\downarrow$FTA\@\xspace}} 
\newcommand*{\BUFTA}{\@ifnextchar{.}{$\uparrow$FTA}{$\uparrow$FTA\@\xspace}} 
\newcommand*{\TDDFTA}{\@ifnextchar{.}{$\downarrow$DFTA}{$\downarrow$DFTA\@\xspace}} 
\newcommand*{\BUDFTA}{\@ifnextchar{.}{$\uparrow$DFTA}{$\uparrow$DFTA\@\xspace}} 
\newcommand*{\FTA}{\@ifnextchar{.}{FTA}{FTA\@\xspace}} 
\newcommand*{\eg}{\@ifnextchar{.}{e.\,g}{e.\,g.\@\xspace}} 
\newcommand*{\ie}{\@ifnextchar{.}{i.\,e}{i.\,e.\@\xspace}} 
\newcommand*{\wrt}{\@ifnextchar{.}{w.r.t}{w.r.t.\@\xspace}} 

\newcommand*{\OBdA}{\@ifnextchar{.}{O.\,B.\,d.\,A}{O.\,B.\,d.\,A.\@\xspace}} 
\newcommand*{\oBdA}{\@ifnextchar{.}{o.\,B.\,d.\,A}{o.\,B.\,d.\,A.\@\xspace}} 
\newcommand*{\usw}{\@ifnextchar{.}{usw}{usw.\@\xspace}} 
\renewcommand*{\dh}{\@ifnextchar{.}{d.\,h}{d.\,h.\@\xspace}} 
\newcommand*{\zB}{\@ifnextchar{.}{z.\,B}{z.\,B.\@\xspace}} 
\newcommand*{\idR}{\@ifnextchar{.}{i.\,d.\,R}{i.\,d.\,R.\@\xspace}} 
\newcommand*{\bzw}{\@ifnextchar{.}{bzw}{bzw.\@\xspace}} 
\newcommand*{\s}{\@ifnextchar{.}{s}{s.\@\xspace}} 
\newcommand*{\su}{\@ifnextchar{.}{s.\,u}{s.\,u.\@\xspace}} 
\newcommand*{\iZ}{\@ifnextchar{.}{i.\,Z}{i.\,Z.\@\xspace}} 
\newcommand*{\iW}{\@ifnextchar{.}{i.\,W}{i.\,W.\@\xspace}} 
\newcommand*{\ua}{\@ifnextchar{.}{u.\,a}{u.\,a.\@\xspace}} 
\newcommand*{\iA}{\@ifnextchar{.}{i.\,A}{i.\,A.\@\xspace}} 
\makeatother 

\newcommand{\Ip}{In particular,\xspace} 
\newcommand{\ip}{in particular,\xspace}

\renewcommand{\phi}{\varphi} 
 
\usepackage{prettyref} 
\newcommand{\prref}[1]{\prettyref{#1}} 
\newrefformat{thm}{Theorem~\ref{#1}} 
\newrefformat{lem}{Lemma~\ref{#1}} 
\newrefformat{def}{Definition~\ref{#1}} 
\newrefformat{cor}{Corollary~\ref{#1}} 
\newrefformat{prop}{Proposition~\ref{#1}} 
\newrefformat{sec}{Section~\ref{#1}} 
\newrefformat{kap}{Chapter~\ref{#1}} 
\newrefformat{ex}{Example~\ref{#1}} 
\newrefformat{eq}{Equation~(\ref{#1})} 
\newrefformat{rem}{Remark~\ref{#1}} 
\newrefformat{fig}{Figure~\ref{#1}} 
\newrefformat{par}{Paragraph~\ref{#1}} 

\newcommand{\intreg}{\ensuremath{\mathit{int_{\mathrm{Reg}}}}}

\title{Properties of Graphs Specified by a Regular Language}

\author{Volker Diekert}
{Universit\"at Stuttgart, Formal Methods in Informatics, Germany \and \url{https://fmi.uni-stuttgart.de/ti/team/diekert/}}
{diekert@fmi.uni-stuttgart.de}
{https://orcid.org/0000-0002-5994-3762}{}

\author{Henning Fernau}
{Universit\"at Trier, Fachbereich IV, Informatikwissenschaften, Germany \and \url{https://www.uni-trier.de/index.php?id=49861}}
{fernau@informatik.uni-trier.de}
{https://orcid.org/0000-0002-4444-3220}
{}

\author{Petra Wolf}
{Universit\"at Trier, Fachbereich IV, Informatikwissenschaften, Germany \and \url{https://www.wolfp.net/}}
{wolfp@informatik.uni-trier.de}
{https://orcid.org/0000-0003-3097-3906}
{DFG project FE 560/9-1}

\authorrunning{V. Diekert et al.}

\Copyright{Volker Diekert et al.} 



\hideLIPIcs  

\EventEditors{John Q. Open and Joan R. Access}
\EventNoEds{2}
\EventLongTitle{42nd Conference on Very Important Topics (CVIT 2016)}
\EventShortTitle{CVIT 2016}
\EventAcronym{CVIT}
\EventYear{2016}
\EventDate{December 24--27, 2016}
\EventLocation{Little Whinging, United Kingdom}
\EventLogo{}
\SeriesVolume{42}
\ArticleNo{23}
\date{Otober 11th, 2021}
\begin{document} 
\maketitle
 \begin{abstract} 
 Traditionally, graph algorithms get a single graph as input, and then they should decide if this graph satisfies a certain property~$\Phi$.
What happens if this question is modified in a way that we get a possibly infinite family of graphs as an input, and the question is if there is a graph satisfying~$\Phi$ in the family? We approach this question by using formal languages for specifying families of graphs, in particular by regular sets of words. We show that certain graph properties can be decided by studying the syntactic monoid of the specification language $L$ if a certain torsion condition is satisfied. This condition holds trivially if $L$ is regular. 
More specifically, we use a natural binary encoding of finite graphs over a binary alphabet $\Sigma$, and we define a regular set $\mathbb{G}\subseteq \Sigma^*$ such that every nonempty word $w\in \mathbb{G}$ defines a finite  and nonempty graph. 
Also, graph properties can then be syntactically defined as languages over $\Sigma$.
Then, we ask whether the automaton $\mathcal{A}$ specifies some graph satisfying a certain property~$\Phi$.
Our structural results show that we can answer this question for all ``typical'' graph properties.

In order to show our results, we split $L$ into a finite union of subsets and every subset of this union defines in a natural way a single finite graph $F$ where some edges and vertices are marked.  The marked graph in turn defines an infinite graph $F^\infty$ and therefore the family of finite subgraphs of $F^\infty$ where $F$ appears as an induced subgraph.
This yields a geometric description of all graphs specified by $L$
based on splitting $L$ into finitely many pieces; then using the notion of graph retraction, we obtain an easily understandable description of the graphs in each piece.
\end{abstract}

\section*{Preamble}
The conference abstract of the present paper appeared in \cite{DiekertFW_DLT21}.
Here, we give full proofs and we correct some mistakes. 
\section{Introduction}\label{sec:intro}  
The paper is about families of finite graphs specified by a regular language, and their properties. 
When dealing with algorithms, a graph is often specified by its adjacency matrix or by its induced edge-list together with the number of isolated vertices, if there are any. In either representation, a graph comes with a linear order on the vertices and the edges are directed. Moreover, an adjacency matrix  ignores multiple edges, but self-loops may occur.
We follow these conventions in our paper.
We encode\footnote{We briefly discuss our encoding of graphs as words (and some related work)
in \prref{sec:GrRW}.} a finite graph 
$G=(V,E)$ as a word over the binary alphabet $\Sigma=\os{a,b}$ as follows: the $i$-th vertex $u_i$ of a 
  graph is encoded by $ab^ia$ and the edge $(u_i,u_j)$  is encoded by  $ab^ia a ab^ja$.
Thus, every word $w$ in $\G=\os{ab^+a,ab^+aaab^+a}^+$ represents in a natural way a unique graph $\rho(w)$ because $\os{ab^+a,ab^+aaab^+a}$ is a regular code.  
Given a finite graph $G=(V,E)$ with any linear order on the vertices, we obtain a code word $\gam(G)$ in $\G$ as follows. We write $V$ as $\os{1\lds |V|}$ using an arbitrary but fixed bijection, and then we write the edges and the isolated vertices
in the order which yields the short-lex normal form of $G$ in $\oi\rho(G)$. This means that first, all edges are listed and then, possible isolated vertices follow.
We are interested in abstract graphs, only. Thus, isomorphic graphs are treated as equal. Therefore, several $\gam(G)$'s are possible, although even then, the short-lex normal form would give a unique syntactic representation of $G$ if necessary.

We cannot avoid that every nonempty graph $G$ has infinitely many representations $w\in \G$ such that $\rho(w)=G$.
For example, the one-point graph $(\os{\star},\es)$ is represented 
by all words in the regular set $L_i=(ab^ia)^+$ as soon as $i\geq 1$, i.e., for all $w\in L_i$ we have $\rho(w)=G$.

Given any $L\subseteq \G$, it defines a set of graphs $\rho(L)$. 
The main interest is when $\rho(L)$ is infinite but $L\subseteq \G$ is regular. The aim is to ``understand'' the infinite set 
of graphs in $\rho(L)$. It is far from obvious that this is possible. If 
$L$ is finite, then  $\rho(L)$ is finite, too. But the converse is false.
As we will see, if $L$ is regular, then we can decide finiteness of $\rho(L)$; and if $\rho(L)$ is finite, then we can compute all its graphs. If however, $\rho(L)$ is infinite, then a global understanding of $\rho(L)$ is, a priori, not easy. 

\subsection{A sketch of our approach and our results}

Let us try to give a high-level explanation of the underlying geometric idea how to approach $\rho(L)$. 
Remotely, it is like understanding the geometry of a topological manifold using the fact that it locally resembles an Euclidean space. For example, it is possible to realize a torus (which is a compact $2$-dimensional surface) 
as a unit square where opposite edges are identified. Every point has on open neighborhood which looks like $\R^2$ and from that one easily derive that the so-called \emph{fundamental group} (which is a global property) is the group $\Z\times \Z$. Therefore we cannot transform a torus neither into a sphere nor into a soup tureen with two or more handles.  

In our case, we deal with purely combinatorial objects. Nevertheless, we wish to understand the set of graphs $\rho(L)$ by constructing a finite subset of graphs together with an ``open neighborhood'' around these graphs such that $\rho(L)$ is covered by that construction. Thus, if we want to check whether a certain property $\Phi$ is satisfied by some graph in $\rho(L)$, then it is enough that we are able to check that locally. The key idea is to cut first 
$L$ into pieces using the algebraic property that a regular language  $L$ has a finite syntactic monoid $M_L$. Hence, $L$ is a finite union of congruence classes; and we obtain an important saturation property:
whenever $w\in L$, then we define the set of words $[w]$ to be all words in the same congruence class of $w$. So, $\rho([w])$ plays the role of an
open neighborhood around the graph $\rho(w)$. Inside each $\rho([w])$, we define finitely many ``smallest''  graphs. Thereby, we find 
 a finite set of finite graphs~$\cF$ such that the collection of these finitely many graphs still has the entire information about $\rho(L)$. In order to reveal that information, we construct for each $F\in \cF$ a (possibly infinite) graph 
$F^\infty$. The graph $F^\infty$ contains $F$ as an induced subgraph; and $F^\infty$ comes 
with a graph morphism onto the graph $F\in \cF$. 
The structure of that infinite graph is fully explicit and actually easy to understand. For example, it might happen that 
$F$ consists of a single edge between two endpoints and $F^\infty$
is the complete infinite bipartite graph $(\N,\N, \N\times\N)$. 
Our result shows that, for every regular language~$L$,  we have $G\in \rho(L)$ \IFF  for some $F$, $G\in \rho(L)$ appears as a finite subgraph of $F^\infty$ containing $F$.

Our geometric approach to $\rho(L)$ has two steps. First, we use the algebraic notion of syntactic monoid. The second step is a graph theoretical definition of $F^\infty$. 
The outcome of following this road map is \prref{cor:4cls}.
It tells us that (with respect to our positive and negative decidable results) it is enough to consider only four different classes $\cC_1\subset \cdots \subset \cC_4$ of graphs~$\rho(L)$.
\begin{enumerate}
\item $\rho(L)\in \cC_1$ \IFF the set $\rho(L)$ is  finite.
\item $\rho(L)\in \cC_2$ implies that $\rho(L)$ has bounded tree-width.
\item $\rho(L)\in \cC_3$ implies that every connected finite bipartite graph appears as a connected component of some $G\in \rho(L)$.
\item $\rho(L) \in\cC_4$ implies that every connected finite graph appears as a connected component of some $G\in\rho(L)$.
\end{enumerate}
Moreover, if $L$ is regular, then  we can compute 
  the largest $\ell$ such that
  $\rho(L)\in \cC_\ell$.
Caveat:  
it may happen that $\rho(L)$ is in $\cC_3$ and in addition it contains arbitrarily large connected non-bipartite graphs, but nevertheless $\rho(L)\notin\cC_4$.

  Since the syntactic monoid of a regular language is finite, we find some 
  $t,p\in \N$ with $p\geq 1$, \emph{threshold} and \emph{period}, such that for every $n\in \N$ there is some 
  $c\leq t+p-1$ with $b^c \equiv_L b^n$ where $\equiv_L$ denotes the \emph{syntactic equivalence}. 
The pair $(t,p)$ tells us that $b^c \equiv_L b^n$ implies first, $n=c$ for  all $0\leq c < t$ and second, $b^n \equiv_L b^{n+p}\iff n\geq t$. This is the key observation when proving that we have no more than these four classes above. If $L\sse \G$ is not regular, then the syntactic monoid $M_L$ is infinite. Still there are interesting examples where $M_L$ satisfies the \emph{Burnside condition} that  all cyclic submonoids of $M_L$ are finite. If so, then there exist $t,p\in \N$ with $p\geq 1$ such that the syntactic properties stated above hold for the powers of the letter~$b$. In this case, we say that $L$ satisfies the \emph{\btptorp}. \prref{thm:allreg} shows that for every subset  $L\sse \G$ satisfying the $(b,t,p)$-torsion property, there exists a regular set $R\sse \G$ such that $\rho(L)=\rho(R)$. This is quite an amazing result. Its proof relies on the fact that  $\rho(L)$ is determined once we know the Parikh-image of $\rf(L)$ in $\N^{t+p-1}$, where for $w\in \G$, the  \emph{reduced form} $\rf(w)$ is obtained by replacing every $b^n$ by $b^c$, where $c$ is the smallest $0\leq c \leq t+p-1$ such that $b^c \equiv_L b^n$. Hence, for deciding whether some graph $G\in \rho(L)$ satisfies a property, we can assume that $L$ is regular. 
We are interested in decidable properties~$\Phi$, only. Thus, we assume that the set $\set{G \text{ is a finite graph}}{G\models \Phi}$ is decidable. If $\rho(L)$ is finite, then we can compute all graphs in $\rho(L)$ and we can output all $G\in \rho(L)$ satisfying~$\Phi$. 

Finiteness of $\rho(L)$ is actually quite interesting and important. It is a case where a representation of~$L$ by a DFA or a regular expression can be used for data compression. 
The minimal size of a regular expression  (or the size of a DFA) for~$L$ is never worse than listing all graphs in $\rho(L)$, but it might be exponentially better.
For a concrete case, we refer to \prref{ex:corona}. 
The compression rate becomes even better if we use a context-free grammar which produces a finite set $L$ of words in $\Sigma^*$, only.  In the extreme case, $L$ is a single word $w$. Then it might happen that the grammar (or \emph{straight-line program}) is exponentially more succinct than writing $w$ as a word in $\Sigma^*$. Thus, possibly we can decide the existence of a graph in $\rho(L)$ satisfying $\Phi$ even though $L$ is highly compressed by the chosen graph representation.

If $L$ is regular, the existence of planar graphs in $\rho(L)$ is 
conceptually easy to decide: Given~$L$, we can can compute a number 
$n(L)\in \N$ such that $\rho(L)$ contains a planar graph
\IFF there is some $w\in L$ of length at most $n(L)$ such that
$\rho(w)$ is planar. On a meta-level, whenever 
we were able to decide whether some $G\in \rho(L)$ satisfies~$\Phi$, then we found effectively a corresponding number $n(L)$.
Moreover, positive decidability results are easy to establish for typical graph properties like ``planarity'' and many other graph properties. 

The second class $\cC_2$ implies  that $\rho(L)$ has bounded tree-width.  In this case, by \cite{CourcelleHandbookGG97,CourcelleEngelfriet2012,Seese91} we know that given any property $\Phi$ which is definable in Monadic Second-Order logic, \emph{MSO} for short, then it is decidable 
  whether there is a graph in  $\rho(L)$ satisfying $\Phi$.
  Languages $L\sse \G$ such that first, $\rho(L)$ has finite tree-width and second, the \btptorp holds 
  can be visualized as a set of graphs sharing some finite subgraph as a backbone structure to which arbitrarily large stars can be glued.
  This observation leads to \prref{thm:endmark}: The satisfiability problem for MSO-sentences is decidable 
  for language in the second class. Not surprisingly, the proof of \prref{thm:endmark} uses  
\prref{thm:allreg}.  

  For the other two classes, the picture changes drastically:  the First-Order theory
  (\emph{FO} for short) becomes undecidable~\cite{Tra50}.  Conversely, we are not aware of any ``natural'' graph property~$\Phi$ 
  (which is not encoding Turing machine computations)
where the satisfiability problem for $\Phi$ is not  trivial for $\cC_3$ and $\cC_4$. 
For example, for these classes $\rho(L)$ contains non-planar graphs, because $\rho(L)$ contains a graph where the complete bipartite graph $K_{2,3}$ is a connected component. For a similar reason, $\rho(L)$ contains graphs without any perfect matching. 
It is therefore more interesting to know whether some $G\in \rho(L)$
allows a perfect matching. This problem is decidable, as we show, but the decision procedure is more involved.

\subsection{Encoding of graphs and related work}\label{sec:GrRW}
Our encoding of graphs by using words over a binary alphabet 
is quite natural but obviously not unique. 
For instance, one could use larger alphabets, say, a unique letter per vertex, in writing down vertex or edge lists. 
As we use a code to write down vertex and edge names, we could interpret our encoding also as using a larger alphabet.
However, in the context of the questions that we discuss in this paper, this would lead to automata over infinite alphabets,
and we wanted to avoid discussing these here.

The bit complexity of $\gam(G)$ (encoding an edge-graph~$G$ with $n$ vertices and $m$ edges) is $\Oh(n\cdot m)$ and hence as good as traditional incidence matrices. More compact representations seem to lead to encodings that are not fit to be tested by finite automata and are hence avoided.


With the idea of using larger alphabets, still completely different encodings are possible. For instance, Kitaev and Seif introduced in~\cite{KitSei2008}  a representation of directed acyclic graphs by associating vertices to sets of letters of a word. This is also interesting for our discussions, as \cite[Thm.~1.8]{KitSei2008})  yields a characterization of the Word Problem of Perkins's semigroup $\mathbf{B}_\mathbf{2}^\mathbf{1}$ in terms of  graph problems. In~\cite{BerMah2016}, again different interpretations of words as graphs and also typical graph problems are investigated for these encodings. Also, Bera and Mahalingam~\cite{BerMah2016} draw connections to Parikh images.

In~\cite{KuskeDLT21}, Kuske generalizes results of de~Malo and de~Oliveira Oliveira in 
\cite{MeloO_CSR2020} on \emph{Second-Order Finite Automata} by using automatic structures.
As an application, Kuske shows in~\cite[Thm.~3.6]{KuskeDLT21} how to decide typical properties of languages classes accepted by second-order finite automata.

Although our results go beyond regular sets $L$, the focus and the motivation comes from a situation when $L$ is regular.  A typical question could be whether there exists some planar graph in $\rho(L)$. Solving this type of decision problems was the motivation to study \emph{regular realizability problems} in \cite{AndersonLRSS09,VyalyiR15} and, independently, 
calling them $\intreg$-problems\footnote{The notation $\intreg$ refers to \emph{intersection non-emptiness with regular} languages.}  in \cite{GulerKLW18,Wolf2019,Wolf2020}. 

\section{Notation and preliminaries}\label{sec:nota}
Some of the following notation was introduced and explained in the introduction, \prref{sec:intro}. For convenience, we repeat them. 
We let $\N= \os{0,1, 2, \ldots}$ be the set of natural numbers and 
$\N_{\infty}= \N \cup \os{\infty}$.
Throughout, if $S$ is a set, then we identify a singleton set $\os x\sse S$ with the element $x\in S$. The power set of $S$ is identified with $2^S$ (via characteristic functions). 
If $E\sse X\times Y$ is a relation, then $\oi E$ denotes its inverse relation
$\oi E=\sset{(y,x)\in Y\times X}{(x,y)\in E}$. By $\id_X$, we mean the
identity relation. Recall that $Y^X$ denotes the set of mappings from a set $X$ to a set $Y$.
 If $f:X\to Y$ and $g:Y\to Z$ are mappings, then $gf:X\to Z$ denotes the mapping defined by $gf(x)=g(f(x))$. If convenient, we abbreviate $f([x])$ as $f[x]$. 
 
Henceforth, $\Gam$ denotes a finite alphabet. For a word $w=a_1\cdots a_n\in \Gam^*$ with $a_i\in \Gam$, we let  $\ov w=a_n\cdots a_1$ be the \emph{reversal} of~$w$. That is, we read the word $w$ from right to left. 

Each alphabet is equipped  with a linear order on its letters\footnote{This convention is standardized in DIN 44300 and ISO 2382.}.  
The linear order on $\Gam$  induces the \emph{short-lex linear order} $\slex$ on $\Gam^*$. 
 That is, for $u,v\in \Gam^*$, we let $u \slex v$ if either $\abs u <\abs v$ or  
 $\abs u =\abs v$, $u=pcu'$, and $v=pdv'$ where $c,d\in \Gam$ with $c< d$. Here, $\abs u$ denotes the length of $u$. Similarly, $|u|_c$ counts the number of occurrences of letter $c$ in~$u$. 
 We also fix the notation $\Sigma=\os{a,b}$ with $a\neq b$. 
 
\subsection{Monoids}\label{sec:mons}
A monoid $M$ is a semigroup $(M,\cdot)$ with a neutral element~$1\in M$. 
If we use a multiplicative notation, then $1$ denotes the neutral element of a monoid. \Ip the empty word in free monoids is denoted by $1$ as well. 
In commutative monoids, we might use an additive operation, and then the neutral element is denoted as $0$. This is standard and there will be no risk of confusion. Cyclic monoids are commutative because, by definition, they are generated by a single element. 
Every finite cyclic monoid $M$ is defined by two numbers  
 $t,p\in \N$ with $p\geq 1$ (where $t$ is the \emph{threshold} and $p$ is the  \emph{period}) such that $M$ is isomorphic to the quotient monoid $C_{t,p}$ of $(\N,+,0)$ with the defining relation $t=t+p$. Hence, the carrier set of $C_{t,p}$ equals $ 
 \{0,1,\dots,t+p-1\}$. If $t=0$ and $p=1$, then  
 $C_{t,p}$ is the trivial monoid~$\os 0$.  

If $M$ is a monoid, then $u\leq v$ means in our paper $v\in M uM$. That is,  
$u$ is a \emph{factor} of $v$. This notation applies, \ip to the monoids $\Gam^*$ and  $\N^\Gam$. Here, $\N^\Gam$ denotes the free commutative monoid
over $\Gam$. Since $\Gam$ is finite, $\N^\Gam$
is the set of mappings from~$\Gam$ to $\N$. Its elements are called \emph{vectors}.

\subsubsection{Syntactic monoids, congruences, and the Word Problem}\label{sec:synm}
Every subset $L\sse \Gam^*$ has a \emph{syntactic monoid} $M=M_R$, see for example \cite{eil74}. 
The elements of $M_L$ are the congruence classes $[u]=\sset{v\in \Sigma^*}{v\equiv_L u}$ \wrt the \emph{syntactic congruence}~$\equiv_L$ which is defined by the following equivalence. 
\begin{align*}
u\equiv_L u' \text{\; \IFF\; } \forall x,y\in \Sigma^*:\, xuv \in L \iff xu'v \in L
\end{align*}
If $L$ is regular, then $M_L$ is finite. 
Later, we do not need that $M_L$ is finite, but we will relax this condition.  It will suffice that the letter $b$ appears in $\Gam$ and is generating a finite submonoid.\footnote{According to \prref{sec:fixR}, we will call this the \btorp.} 

Let $\phi:\Gam^* \to G$ be a surjective \hom onto a finitely generated group $G$. 
Then the \emph{Word Problem} of~$G$ 
denotes the  set $\WP(G)= \set{w\in \Gam^*}{\phi(w)=1}$. If this set is decidable, then we say that Word Problem of~$G$ is decidable because on input $u,v\in \Gam^*$ we can decide whether $\phi(u)=\phi(v)$. 
It is a classical fact (and an easy exercise) that decidability of 
Word Problem does not depend on the generating set and that 
the syntactic monoid of $\WP(G)$ is the group~$G$ itself,
\cite{anisimov71}.  

\subsubsection{Burnside groups}\label{sec:bugr} Recall that $|\Sigma|=2$. The \emph{free Burnside group} $\cB(2,p)$ is defined 
as the quotient
\begin{align*}
\cB(2,p)= \Sigma^*/\set{x^p=1}{x\in \Sigma^*}
\end{align*}
where $p\geq 1$. It is a group, because every $x$ has the inverse  
   element $x^{p-1}$ thanks to $p\geq 1$.   
   For $p$ large enough, Adjan has shown in the 1970s that $\cB(2,p)$ is infinite, answering a question of Burnside dating back in its original form  to 1902. Actually, Adjan also showed the decidability 
   of the Word Problem  of $\cB(2,p)$ if $p$ is large enough. 
Here, a group (with two generators) is called \emph{$p$-periodic} if it is the homomorphic image of some $\cB(2,p)$. 
  Kharlampovich constructed in \cite{KharlampovichJA1995} a periodic group  $B(2,p)$ with a generating set $\Sigma$ and a finite set  
   of words $w_1\lds w_r\in \Sigma^*$  
   such that the group $B(2,p)$ has the monoid presentation 
    $B(2,p)= \Sigma^*/\set{w_i=1 , w^p=1}{1\leq i \leq r\wedge w\in \Gam^+} $
    (as an abstract group) and where the Word Problem  $\WP(B(2,p))$ is undecidable. Thus, the language  $L=\WP(B(2,p))$ is undecidable, 
    nevertheless $w^p\equiv_L 1$ for all $w\in \Sigma^*$.  
 We use this example to illustrate that there are undecidable
 languages satisfying the \btorp which will be defined in \prref{sec:fixR}.

\subsection{Parikh-images}\label{sec:parikh} 
 If $v,w\in \Gam^*$, then $|w|_v$ denotes the number how often 
 $v$ appears as a factor in~$w$, \ie, 
 $|w|_v= \abs{\set{u\in \Gam^*}{\exists s: uvs=w}}$. 
 If $P\sse \Gam^*$, then the \emph{Parikh-mapping \wrt $P$} is  
 defined by $\pi_P:\Gam^*\to \N^{P}$,  mapping
 a word $w\in \Gam^*$ to its \emph{Parikh-vector} $(|w|_v)_{v\in P}\in \N^{P}$. 
 The classical case is $P=\Gam$;  then the Parikh-vector 
 becomes $(|w|_a)_{a\in \Gam}$ and the Parikh-mapping is the canonical \hom from the free monoid  $\Gam^*$ to the free commutative monoid~$\N^\Gam$. As usual, $\N^\Gam$ is partially ordered such that 
\begin{align*}
u\leq v \iff \forall z\in \Gam:\;u(z)\leq v(z).
\end{align*}
Subsets $L\sse \N^\Gam$ which can be written as 
$L= q + \sum_{i\in I}\N p_i$ are called \emph{linear} and a finite union of linear sets is called  \emph{semi-linear}. Use the following classical results. 
\begin{proposition}\label{prop:gs66parikh66}
\begin{enumerate}[(a)]
\item The complement of a semi-linear set $L\sse \N^\Gam$ is effectively semi-linear. Hence, the family on semi-linear sets is an effective Boolean algebra, see~\cite{gs66}.  
\label{gs66}
\item The Parikh-image of a \cfree language is effectively semi-linear,  see~\cite{parikh66}.
\label{parikh66}
\end{enumerate}
\end{proposition}

A subset $S\sse \N^\Gam$ is called \emph{positively downward-closed} if first $v(z)\geq 1$ for all $v\in S$, $z\in \Gam$ and second,  $u\leq v\in S$ and $u(z)\geq 1$ for all $z\in \Gam$ imply $u\in S$. The complement of a positively downward-closed set $S\sse \N^\Gam$ is \emph{upward-closed}, \ie, $u\geq v\in S$ implies $u\in S$. 
An upward-closed 
set $S$ is determined by its set $\min(S)$ of minimal elements. Dickson's Lemma says that the set 
$\min(S)$ is finite for all $S\sse \N^\Gam$. Hence, every 
 upward-closed subset is semi-linear. It follows by \prref{prop:gs66parikh66}~(\ref{gs66}) that every  
positively downward-closed set $S\sse \N^\Gam$ is semi-linear, too. 
This observation is crucial for proving  \prref{thm:allreg}.  
  
\subsection{Graphs}\label{sec:graph} 
 All graphs are assumed to be (at most) countable, given as a pair 
 $G=(V,E)$ where $E\sse V\times V$. An \emph{undirected graph} is the special case where $E=\oi E$, so that $E$ describes the adjacency relation. If $G=(V,E)$ is a directed graph, then $G$ also defines the \emph{undirected graph} $(V,E\cup \oi E)$; 
 and it defines the  \emph{undirected graph without self-loops}  $(V,(E\cup \oi E)\sm \id_V)$. A graph without  isolated vertices is called an \emph{edge-graph}. Hence, the set of edges determines an edge-graph.  
 If $G'=(V',E')$  and $G=(V,E)$  
 are graphs such that $V'\subseteq V$ and $E'\subseteq E$, then  
 $G'$ is  a \emph{subgraph} of graph $G$ and we denote this fact by $G'\leq G$. If $U\sse V$ is any subset, then $G[U] =  
 (U,E\cap U\times U)$ denotes the \emph{induced subgraph} of $U$ in $G$. \emph{A subset $U$ is called \emph{independent} if $G[U]$ is without any edge.}
  A \emph{graph morphism} $\phi: (V',E')\to (V,E)$ is given by a mapping  
 $\phi: V'\to V$ such that  
 $(u,v)\in E'$ implies $(\phi(u),\phi(v))\in E$. If $(V',E')$ and $(V,E)$ are undirected graphs without self-loops,  
 then $\phi: (V',E')\to (V,E)$ is a graph morphism as soon as $(\phi(u),\phi(v))\in E\cup \id_V$.  
 We say that $\phi$ is a \emph{projection} if  
 $\phi$ is surjective on vertices and edges, i.e., $\phi(V')=V$ and $\phi(E')=E$.  
 We consider graphs up to isomorphism, only. Hence, writing $G=G'$ means that graphs $G$ and $G'$ are isomorphic.  
 According to the following \prref{sec:retract}, a graph $F=(V,E)$ is a \emph{retract} of a graph $F'=(V',E')$ if there are \morph{s} $\phi: F'\to F$  and $\gam: F\to F'$ such that 
 $\phi\gam$ is the identity on vertices and edges of $(V,E)$.  
 Hence, $F$ appears in $F'$ as the induced subgraph $F'[\gam(V)]$. 
Another way to say this is that $F$ is an induced subgraph of $F'$ and there is a  \morph $\phi: F'\to F$ which is the identity on $F$.

We consider several special graphs (and graph properties) in our paper. By a \emph{star}, we denote a graph $(V,E)$ such that  
  there exists a vertex $z\in V$ with the property $E=\sset{(z,s)}{s\in V\sm \os z}$. Thus, a star has a \emph{center} $z$ and the directed edges are the outgoing \emph{rays} of the star. We also use this notion to refer to an undirected connected graph where all but possibly one vertex have degree one.
Let $G=(V,E)$ be an undirected graph. $G$ is called a clique if all possible edges (apart from self-loops) are contained in $E$.
$G$ is called bipartite  if $V$ can be partitioned into $V_1$ and $V_2$ such that there are no edges between vertices of the same class of the partition. A bipartite graph is called complete if no further edges can be added without violating bipartiteness.
By $K_n$ we denote a clique with $n$ vertices, and $K_{n,m}$ denotes the complete bipartite graph with $n$ vertices in one of the classes and $m$ vertices in the other one. A vertex set is \emph{independent} if its vertices are pairwise non-adjacent. A vertex set~$C$ is a \emph{vertex cover} if for each edge, at least one endpoint belongs to $C$. 

Some of our results revolve \emph{tree-width}. The notation is due to Robertson and Seymour and was one of the cornerstones 
in their famous graph minor project. A formal definition of tree-width is, for example, in \cite{diestel2012} and in many other  modern textbook on graph theory or graph algorithms. 

In our paper, every word $w\in \Sigma^*$ represents a directed finite graph $\rho(w)=(V(w),E(w))$ together with a linear order on vertices as follows. 
 \begin{align*} 
  V(w)&=\sset{ab^ma\in ab^+a}{ab^ma\leq w}\\ 
  E(w)&=\sset{(ab^ma,ab^na)\in ab^+a\times ab^+a}{ab^maaa b^na\leq w} 
 \end{align*} 
 The empty word represents the empty graph: there are no vertices and no edges.  We extend $\rho$ to $2^{\Sigma^*}$ by $\rho(L)=\{\rho(w)\mid w\in L\}$. 
 Vice versa, if  $G=(V,E)$ denotes a finite graph with a linear order on its vertices, then, for $1\leq i,j\in \N$, the $i$-th vertex is represented by the factor $ab^ia$, and an edge from the $i$-th vertex to the $j$-th vertex is represented 
 by the factor $ab^iaaab^ja$. Thus, vertices are encoded  
 by elements in the set $\V=\sset{ab^ia}{1\leq i \in\N}$ and  
 edges are encoded  
 by elements is the set $\E=\sset{ab^iaaab^ja}{1\leq i,j \in \N}$.  
 Note that $\V\cap\E=\es$ and $\V\cup\E$ is an infinite regular code. Using these conventions, the regular set $\G=(\V\cup\E)^*$ as well as its subset $\E^*\V^*$ represents all finite graphs. 
The same property holds for the complement $\Sigma^*\setminus\G$: it
represents all finite graphs, too. Indeed, $\G \cap a\G=\es$ but $\rho(\G )= \rho(a\G)$, since the words $w\in \G$ and $aw\notin \G$ share the same set of factors from $\V\cup\E$. 
 In contrast to $\G$, infinitely many nonempty words in $\Sigma^*\setminus\G$ represent the empty graph, for example all words without any $b$ or with at most one $a$.
The set $\E^*$ represents all edge-graphs, \ie, all graphs without isolated vertices.  
 Every nonempty finite graph has infinitely many representations in $\G$. 
 For example, there are uncountably many subsets $L\sse (aba)^+\sse \V^+$ and each $\rho(L)$ represents nothing but the one-point graph without  self-loop.  
 In order to choose a unique (and minimal) representation for  
a finite graph $G=(V,E)$, we choose the minimal word $\gam(G)= u_1\cdots u_m v_1\cdots v_n\in \G$ in the short-lex ordering  
 on $\Sigma^*$ such that $\rho\gam(G)=G$, $u_k\in \E$  
 for $1\leq k \leq m$ and $v_\ell\in \V$  
 for $1\leq \ell \leq n$. Each $u_k$ is of the form  
 $ab^iaaab^j a$ representing an edge and each $v_\ell$ is of the form  
 $ab^ia$ representing an isolated vertex. We call $\gam(G)$ the \emph{short-lex representation} of~$G$.
Since $\gam(G)$ is minimal \wrt $\slex$,  we have $m=\abs E$ and $n$ is the number of isolated vertices. For a graph without isolated vertices, \ie, an edge-graph, this means that it is given by its edge list.  
 The set of all $\gam\rho(\G)$ is context-sensitive but not context-free. 
The $uvwxy$-Theorem (\ie, the context-free pumping lemma) does not hold for 
$\gam\rho(\G)$. For instance, all edge-graphs in $\gam\rho(\G)$ with $n$ vertices must be represented by vertex names $ab^ia$ with $1\leq i\leq n$, admitting no `holes' in this vertex name interval, as otherwise there would be a smaller short-lex representation of some graph. Such a property is not maintained by pumping.

\subsection{Retractions and retracts}\label{sec:retract} 
 Let $\rho:X\to Y$ and $\gam:Y\to X$ be mappings between sets.  (This holds more general for mappings which are \morph{s} in some category.) 
If $\rho(\gam(y))=y$ for all $y\in Y$, then $\rho$ is called a \emph{retraction} and $Y$ is called a \emph{retract} of~$X$ with \emph{section}~$\gam$. We also say that 
 $\oi \rho(y)$ is the \emph{fiber} of $y\in Y$.  
 For example, if $\rho:X\to Y$ is a \hom of groups $X$ and $Y$ and $H=\ker(\rho)$ is the kernel, then $\rho$ is a retraction 
 \IFF  $X$ is a semi-direct product of $H$  by~$Y$.  
Another example comes from formal languages: let $X$ be the set of deterministic finite automata (DFAs) where every state is reachable. Then the minimization process defines a retraction to the set of minimal DFAs. 

{Later in \prref{sec:marked}, we define a \emph{marked graph}
as a triple $(V_F,E_F,\mu)$, where $(V_F,E_F)$ is a finite graph and 
$\mu\sse V_F \cup E_F$ is the set of marked vertices and edges. 
Then, $(V_F,E_F,\mu)\mapsto (V_F,E_F)$ defines a retraction by letting $\gam(V_F,E_F)= (V_F,E_F,\es)$.}

Let $\cG$ be the set of finite graphs, $\rho:\G\to \cG$ be the 
representation of graphs by words, and $\gam:\cG\to \G$ the encoding of a graph by its short-lex normal form. Then $\rho$ is a retraction. Retractions are a main tool to understand $\rho(L)$ if 
$L$ is regular or more general, if $L$ satisfies the \btorp, as defined in the next section.

\section{The $b$-torsion property}\label{sec:fixR} 
 We are interested in properties of graphs which are specified  
 by languages $L\sse \G$. If $L$ can be arbitrary, then  
 we can specify uncountably many families of graphs. So, we cannot expect any interesting and general positive (decidability) results.  
 As a minimal request, we restrict our attention to  
 subsets $L\sse \G$ where membership for $\rho(L)$ is decidable.  
 As a matter of fact, membership for $\rho(L)$ might be decidable 
 although membership for~$L$ is undecidable.  
 As we will see in \prref{cor:allreg}, the following definition yields a sufficient condition that membership for $\rho(L)$ becomes decidable. 
\begin{definition}\label{def:genR} 
  Let $\Gam$ be a finite alphabet containing the letter $b$. A subset $L\sse \Gam^*$ satisfies  
  the \emph{\btptorp} if we have: 
  $ 
  b^{t}\equiv_L b^{t+p}. 
  $ 
  It satisfies the \emph{\btorp} if there are $t,p\in \N$ with $p\geq 1$ such that $L$ satisfies  
  the \btptorp. 
 \end{definition} 
Every regular language $L\sse \Gam^*$ satisfies  
 the $b$-torsion property because the syntactic monoid $M_L$ is finite. 
 The \btorp is a strong restriction if $L$ is not regular. For example,  $K=\sset{a^nb^n}{n\in \N}$ does not have that property since  $b^{k}\equiv_L b^{m}\iff k=m$.
The \cfree language  
 $L=\{wa\ov w\mid w\in \{aba, ab^2a\}^+\}$ is not regular, but it satisfies the  $b$-torsion property for $t=3$ and $p=1$.  
 In this case, $\rho(L)$ is a not very interesting set of a few small graphs. The next example shows that there are (non-regular) \cfree languages satisfying  a \btptorp where $\rho(L)$ is infinite. 
\begin{example}\label{ex:cfbtorp}
Let $C$ be a nonempty finite alphabet and $K\sse C^*$ be \cfree. 
Let $h:C^*\to \E^*$ be a \hom. That is, $h$ is defined
by words $h(c)\in (ab^+aaab^+a)^*$ for $c\in C$. Suppose that $h(K)$ is infinite.
Still, the  set  $\rho(h(K))$ is finite. Indeed, let $t = \max\set{i\in \N}{\exists c\in C:\, ab^ia \leq h(c)}$, then graphs in $\rho(h(K))$ have at most $t$ vertices. 
Let us make the language $h(K)$ 
larger by closing $h(K)$ under rewriting rules $b\to b^{1+p}$. 
\Cfree languages are closed under adding \cfree rewriting rules. Therefore, $h(K)$ is \cfree, too. We obtain 
a new \cfree language $L$ with $h(K)\sse L$ and where $L$ satisfies the \btptorp. We claim that $\rho(h(L))$ is  a very rich and an infinite family of graphs (in contrast to the finite set $\rho(h(K))$). 

We content ourselves to consider the case $p=1$.  For $p=1$ it is rather easy to see that  every non-empty finite edge-graph appears in $\rho(h(L))$: we have $\rho(h(L))=\rho(\E^*)$.
Let $G=(V_G,E_G)\in \rho(\E^*)$ and $m=|E_G|$ be the number of edges.
Since $h(K)$ is infinite, there is  some edge
$e=ab^iaaab^ja$ with $1\leq i, j \leq t$ such that $e$ appears in some $w_f\in h(L)$ at least $m$ times as a factor. 
Then, thanks to $p=1$, we have 
$w=(ab^taaab^ta)^\ell\in L$, where $m\leq\ell$ and $\ell$ is very large. 
The graph $\rho(w)$ is a one-point graph with a self-loop. 
Let us come back the graph $G$. Without restriction, we have
$V_G\sse \os{t\lds \ell}$. 
Now, for each $(u,v)\in E_G$, one after  the other, 
we replace one factor $ab^{t}aaab^{t}a$ in $w$ by the factor 
$ab^{u}aaab^{v}a$. This changes the word $w$, but the new word $w$ still belongs to $L$, again thanks to $p=1$. By creating, if necessary, several copies of the same edge, the procedure yields a word 
$w_G$ such that $w_G\in L$ and $\rho(w_G)=G$. 
\qed
\end{example}

As soon as all cyclic submonoids of $M_L$ are finite, $L$ satisfies the \btorp for all letters $b\in \Gam$. For example, consider the Word Problem of any free Burnside group $\cB(2,p)$. All of them satisfy the \btorp. Almost 
all of the groups $\cB(2,p)$ are infinite and therefore the corresponding Word Problems are 
not regular. If it is not regular, then the Word Problem of  $\cB(2,p)$ is even not \cfree, since a periodic group cannot have any non-trivial free group of finite index by~\cite{ms83}.
 
 For the rest of the paper, if $L\sse \Sigma^*$ satisfies the \btorp, then  $t,p\in \N$, standing for \emph{threshold} and \emph{period},  denote those natural numbers such that the cyclic submonoid generated by the letter $b$ in the syntactic monoid $M_L$ is isomorphic to $C_{t,p}$. That is,  we have $t,p\in \N$ with $p\geq 1$, where $t+p$ is minimal such that 
\begin{align}\label{eq:tp} 
\sset{[b^n]}{n\in \N}= \sset{[b^c]}{0\leq c\leq t+p-1}\,. 
\end{align} 
 Moreover, we assume that $L$ is specified such that on input $n\in \N$, we can compute the value $0\leq c\leq t+p-1$ with $b^n\equiv_L b^c$. This assumption is satisfied if $L$ is regular and specified, say, by some NFA.  
 For $L\sse \G$, we have 
 $[ab^ca] = a [b^c] a$ and $[ab^caaab^da] = a [b^c] aaa[b^d]a$. 
The tacit assumption is important for the next definition to compute, for example, the reduced form according to the next defintion.  
 \begin{definition}\label{def:nfwh} 
  Let $L\sse \G$ satisfy the \btptorp according to 
  \prref{def:genR}.  
  For every $[b^n]$, we define its \emph{reduced form} by    
  $\rf{[b^n]}= b^c$ if $[b^c] = [b^n]$ and $0\leq c\leq t+p-1$. 
  Given $w\in \G$,  we define the  
  \emph{reduced form} $\rf(w)$ by replacing every factor $ab^ma\leq w$ by $a\, \rf{[b^m]}a$. The \emph{saturation} $\wh w$ of $w$ is defined by replacing  
  every factor $ab^ma\leq w$ by the set $a[b^{m}]a$. Hence, $\rf(w)\in \wh w\sse \G$. 
 \end{definition}

 \begin{remark}\label{rem:small} 
  Let $L\sse \G$ satisfy the \btptorp. By possibly decreasing $t$ and/or $p$, we may assume that for every $1\leq c \leq t+p-1$, there is some $w\in L$ such that $ab^ca\leq \rf(w)$. Moreover,  
  we have $[b^c] = \os{b^c}$ \IFF $c<t$.  
\qed 
\end{remark} 
\begin{lemma}\label{lem:nfwh} 
  Let $L\sse \G$ satisfy the \btptorp. 
  Then for every $w\in \G$, 
\begin{align*} 
   w\in L \iff \wh w\sse L \iff \rf(w)\in L\,. 
\end{align*} 
 \end{lemma} 
\begin{proof} 
 Trivial, by definition of the \btptorp. 
\end{proof} 

\section{Main results}\label{sec:mr}
The main results of the paper are: (1) for  $L\sse \G$ satisfying the \btorp (see \prref{sec:fixR}), there is a regular language $R\sse \G$ with $\rho(L)=\rho(R)$ and (2) 
 for a \cfree language satisfying the \btorp (e.g., any regular language) $R\sse \G$, we have 
 an effective `geometric description' of the graphs in $\rho(R)$. From these representations, we can deduce our classification and (un)decidability results as already mentioned in the introduction. The (un)decidability results are detailed in the next section.
 
 This geometric description is obtained as follows. Using the fact that $R$ is regular (or \cfree satisfying the \btptorp), in a first step, we find effectively  a semi-linear description of $\rho(R)$. In a second step, we compute a finite set of finite graphs. Each member $F$ in that finite family is a retraction  of some possibly infinite graph $F^\infty$.   
 The description of each $G\in \rho(L)$ is given by selecting some  
 $F$ and the cardinality of every fiber. The precise meaning will become clear later. As a consequence of the description, we are able to show various decidability results.  

\subsection{Examples}
 The following examples serve as an introduction to a more general situation we will face later.  
\begin{example}\label{ex:hff} 
  In the following, we let $R\sse\G$ and  
  $t,p\in \N$ with $p\geq 1$ such that $b^n\equiv_R b^{n+p}$ for all $n\geq t$. Moreover, we let $1\leq c <t$ such that $[b^c]=\os{b^c}$. 
  \begin{enumerate} 
   \item \label{ex:vd1} 

   Let $w\in R\sse(ab^c aaa b^n(b^p)^* a)^+\sse \G$ with $t \leq n$. This implies $t \leq n<t+p$ and 
   $w$ contains a factor 
   $ab^c aaa b^d a$ with $t \leq d< t+p$  and $[b^d]= b^{d+\N p}$. 
   We have $w\in(ab^c aaa b^n(b^p)^* a)^m$ for $m=|w|_a/5$.  Hence,   
   $w=(ab^caaab^{d_1}a)\cdots (ab^caaab^{d_m}a)$ where  
   $d_i=n+k_ip$ with $k_i \in \N$ for $1\leq i \leq m$. The  
   set $\sset{d_i}{1\leq i \leq m}$ can have any cardinality  
   $s$ in $\os{1\lds m}$.  
   Therefore, $\rho(w)$ is a single star with at least one ray and at most $m$ rays. If $R$ is finite, then $\cF=\rho(R)$ is an effective finite collection of stars with at least one ray and at most $r$ rays where $r=\max\sset{|w|_a/5}{w\in R}$. 
  
   We claim that $\cF$ is infinite \IFF  there is some $M\geq |M_R|$ such that $(ab^caaab^na)^{M}\in R$. Moreover, if $\cF$ is infinite, then $\cF$  
   is the set of all finite stars with at least one ray. 
   The claim holds if $\sup\sset{|w|_a/5}{w\in R}<\infty$, as in this case $\cF$ is finite.  
   Thus, let $\sup\sset{|w|_a/5}{w\in R}=\infty$. Then there is some $w\in R$   
   such that $ab^caaab^ta$ appears at least $|M_R|$-times as a factor. This implies that there is some $M\geq |M_R|$ such that $(ab^caaab^ta)^{M}\in R$. 
   The claim follows.  
   
   One can show that $S=(aba a ab^2b^*a)^*(aba)$ is locally testable and therefore star-free.  Hence, the set of all finite stars is specified by a star-free subset of $\Sigma^*$.   
\item \label{ex:vd2} 
If $R\sse(ab^c aaa b^cb^+ a)^*ab^ca\sse \G$, then $\rho(R)$ is set of stars with center $ab^c a$ and outgoing rays to vertices $ab^d a$ where $d>c$. Moreover, the following dichotomy holds:
   The set of stars in $\rho(R)$ is either finite or it contains almost all finite stars. Indeed, $\rho(R)$ is a set of stars with center $c$, possibly without rays. If $\rho(R)$ is finite we are done. 
   Otherwise, let $\rho(R)$ be infinite. Then, for each $r \in \N$, there is a star in $\rho(R)$ with more than $r$ rays. This implies that for all $r\in \N$ there is some $\ell\geq r$ and a word 
   $w\in R$ which has more than $\ell$ pairwise different factors $ab^c aaa b^{d_i} a$ with $r\leq d_i$. If $r$ is large enough, then each of these factors 
   can be replaced by a factor $ab^c aaa b^{c_i}a$ where we have that 
   $t< c_i\leq t+p$. This yields a word $w'\in R$. Now, $w'$ is very long as $\ell$ is very large. Hence, we can factorize the word 
   $w'=uw''v$ such that first, $w''$ contains one of these factors $ab^c aaa b^{c_i}a$ and second, $S= u (w'')^+v \sse R$. Since $ab^c aaa b^{c_i}a\equiv_R 
   ab^c aaa b^{c_i +pk}a$ for all $k\in \N$, we conclude that $\rho(S)$ contains almost all stars, i.e., all stars but finitely many that are missing. \qed
\end{enumerate} 
 \end{example} 
 \medskip
The \btptorp is trivially satisfied if $L\sse \G$ is a finite set. An interesting case motivated
 by data compression. As mentioned in \prref{sec:intro}: if  $L$ is finite, then the minimal size of a regular expression for $L$ is never worse than listing all graphs in $\rho(L)$, but it might be exponentially better.
 This type of data compression is important and well-known~\cite{LarMof2000,LohManMen2013}. It is used in  DNA-computing and bio-inspired modeling, frequently, and it is the basis of practical algorithms like RePair.
\begin{figure}
\begin{subfigure}[c]{.5\textwidth}\centering
	\tikzset{every loop/.style={}}
\begin{tikzpicture}
	\tikzset{mystyle/.style={state, inner sep=2pt,minimum 	size=1pt, fill=black, node distance=20pt}}
	\node[mystyle] (0) at (0.5,0) {};
	\node[mystyle] (1) at (1.25, 0) {};
	\node[mystyle] (2) at (1.75,0.4){};
	\node[mystyle] (3) at (2, 1) {};
	\node[mystyle] (4) at (1.75, 1.6) {};
	\node[mystyle] (5) at (1.25, 2) {};
	\node[mystyle] (6) at (0.5, 2) {};
	\node[mystyle] (7) at (0, 1.6) {};
	\node[mystyle] (8) at (-.25, 1) {};
	\node[mystyle] (9) at (0,0.4){};
	
	\node[mystyle] (0') at (0.3,-0.5) {};
	\node[mystyle] (1') at (1.45, -0.5) {};
	\node[mystyle] (2') at (2.2,0.1){};
	\node[mystyle] (3') at (2.5, 1) {};
	\node[mystyle] (4') at (2.2, 1.9) {};
	\node[mystyle] (5') at (1.45, 2.5) {};
	\node[mystyle] (6') at (0.3, 2.5) {};
	\node[mystyle] (7') at (-0.45, 1.9) {};
	\node[mystyle] (8') at (-.75, 1) {};
	\node[mystyle] (9') at (-0.45,0.1){};
	\path
	(0) edge (1)
	(1) edge (2)
	(2) edge (3)
	(3) edge (4)
	(4) edge (5)
	(5) edge (6)
	(6) edge (7)
	(7) edge (8)
	(8) edge (9)
	(9) edge (0)
	(0) edge (0')
	(1) edge (1')
	(2) edge (2')
	(3) edge (3')
	(4) edge (4')
	(5) edge (5')
	(6) edge (6')
	(7) edge (7')
	(8) edge (8')
	(9) edge (9')
	;
\end{tikzpicture}
\subcaption{A full crown around a cycle of length $10$.\;}
\label{fig:crown}
\end{subfigure}\begin{subfigure}[c]{.5\textwidth}\centering
	\tikzset{every loop/.style={}}
\begin{tikzpicture}
	\tikzset{mystyle/.style={state, inner sep=2pt,minimum 	size=1pt, fill=black, node distance=20pt}}
	\node[mystyle] (0) at (0.5,0) {};
	\node[mystyle] (1) at (1.25, 0) {};
	\node[mystyle] (2) at (1.75,0.4){};
	\node[mystyle] (3) at (2, 1) {};
	\node[mystyle] (4) at (1.75, 1.6) {};
	\node[mystyle] (5) at (1.25, 2) {};
	\node[mystyle] (6) at (0.5, 2) {};
	\node[mystyle] (7) at (0, 1.6) {};
	\node[mystyle] (8) at (-.25, 1) {};
	\node[mystyle] (9) at (0,0.4){};
	
	\node[mystyle] (0') at (0.3,-0.5) {};
	\node[mystyle] (2') at (2.2,0.1){};
	\node[mystyle] (3') at (2.5, 1) {};
	\node[mystyle] (5') at (1.45, 2.5) {};
	\node[mystyle] (7') at (-0.45, 1.9) {};
	\node[mystyle] (9') at (-0.45,0.1){};
	\path
	(0) edge (1)
	(1) edge (2)
	(2) edge (3)
	(3) edge (4)
	(4) edge (5)
	(5) edge (6)
	(6) edge (7)
	(7) edge (8)
	(8) edge (9)
	(9) edge (0)
	(0) edge (0')
	(2) edge (2')
	(3) edge (3')
	(5) edge (5')
	(7) edge (7')
	(9) edge (9')
	;
\end{tikzpicture}
\subcaption{A mutation with six cusps, only.}
\label{fig:crownvd}
\end{subfigure}
\caption{An $2$-dimensional illustration of \prref{ex:corona}: 
A small regular expression defines a crown with $n$ cusps (in the picture: $n=10$) as well as an exponential number of possible mutations.}
\end{figure}
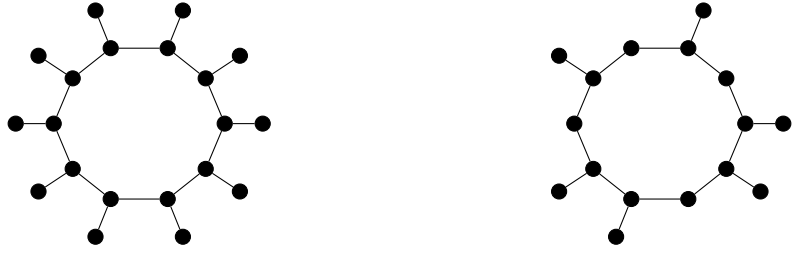

  \begin{example}\label{ex:corona}
Let $G=(V,E)$ be a connected planar graph with vertex set $V=\os{1\lds n}$ with $n\geq 3$ together with an embedding into the $2$-dimensional sphere  For every subset $S\sse \os{n+1\lds 2n}$, let $G_S$ denote the graph 
\(G_S=(V\cup S,E\cup \set{(s,s-n)}{s\in S}).\)
The family $\cC_n=\set{G_S}{S\sse \os{n+1\lds 2n}}$ might contain
exponentially many graphs in $n$. This happens, for example, if $G$ is a cycle of $n$ nodes. Then $\cC_n$ has 
more than  $2^n/2n \in 2^{\OO(n)}$ connected planar graphs. 
If we embed $G$ in the 2-dimensional sphere where the additional edges are spikes pointing out of the sphere, then $G_S$ can be visualized as a discrete model of a
3-dimensional ``crown with at most~$n$ cusps''. A $2$-dimensional representation of a full crown (having all possible cusps) is depicted in \prref{fig:crown}, while 
\prref{fig:crownvd} shows the situation when some cusps were chosen to be removed.
It is straightforward to
write down a $2n$-fold concatenation of finite sets which describes a finite set $L_n\sse \G$ such that $\rho(L_n)=\cC_n$. The size of the corresponding regular expression is $\Oh(n^{2})$. 
Thus, we have a polynomial-size blueprint potentially \emph{producing a family of exponentially many mutations of single 
``corona''}, the Latin word for ``crown''. 
\qed
\end{example}

\subsection{Introducing new alphabets}
After seeing a couple of examples, we introduce certain subsets of $\G$ as alphabets to express subsets of graphs. This
prepares the geometric viewpoint to view a graph as  a point in 
the $d$-dimensional space $\N^d$ if the size of the chosen alphabet is $d$. If $L$ is regular, then the dimension depends on $L$ and it  can be quite large, but it is computable by using the reduced form of words in $L$. Below, all this will be explained in  details.

Let $\ell\in \N$. Depending on $\ell$, we define 
two finite and disjoint sets that we consider as alphabets:
\begin{align}\label{eq:Aell}
   A_\ell &= \sset{ ab^{m}aaab^{n}a }{1\leq  m,n \leq \ell}\sse \E,\\ 
   \label{eq:AellB}
 B_\ell &= \sset{ab^{m}a}{1\leq  m\leq \ell}\sse \V.
  \end{align}
Note that $A_\ell = B_\ell a B_\ell$.  
  By $C_\ell$, we denote the union of $A_\ell$ and $B_\ell$, which is also a finite alphabet with a linear order between letters  given by  
  the following definition:  
\begin{equation}\label{eq:linC} 
   x\leq_\ell y\iff 
   xy\in A_\ell B_\ell \vee (xy\in (A_\ell A_\ell \cup B_\ell B_\ell) \wedge x\slex y). 
\end{equation} 
We have $A_\ell\cap B_\ell=\es$ and $\N^{C_\ell}$ has dimension
$d=\ell(\ell +1)$.
Actually, $C_\ell$ is a code, \ie, 
 $c_1\cdots c_m = d_1\cdots d_n \in \Sigma^*$ with $c_i,d_j\in C_\ell$ implies  
$m=n$ and $c_i=d_i$ for all $1\leq i \leq m$. Hence, $\leq_\ell$ is well-defined by \prref{eq:linC}. 

 \begin{definition}\label{def:AB}  
Let $L\sse \G$ satisfy the \btorp. Then we define 
an alphabet  $C_L\sse C_\ell$  where $\ell$ is large enough such that for every $w\in L$ its reduced form $\rf(w)$ can be written as a word in $C_\ell^*$. To make the definition unique, we choose $C_L$ to be the smallest 
set $C_L\sse \bigcup_{\ell\in \N} C_\ell$ which has this property.
\end{definition}
If the context to $L$ is clear, then we drop the index $L$ and 
we denote by $C$ any subset of some $C_\ell$ such that
 $C_L\sse C \sse C_\ell$. This flexibility is useful if we wish to introduce new vertices or new edges and we need ``fresh'' names in $C$ for them. In order to distinguish the factor ordering 
and the linear ordering defined by \prref{eq:linC}, we denote the latter by $\leq_C$.

The linear order $\leq_C$ on $C$ defines a corresponding 
short-lex ordering  on $C^*$. 
Moreover,  
 if $uxv\in C^+$ with $x\in A$ and $u,v\in \Sigma^*$, then $u,v\in C^*$. The analogue for $y\in B$ does not hold, in general. 
For example, $aba\in B$ and $abaaaba\in C$, but $aaba\notin C$. 
 As $C$ is a code, the inclusion $C\sse \Sigma^*$  yields an embedding
 $h_C:C^+ \to \Sigma^+$. Let $w\in C^*$ and $G=\rho(w)$. 
 If $L= \oi\rho(G)$, then  
 the minimal element in $\oi h_C(L)$ \wrt 
 the short-lex ordering for words in $C^*$ is a word in $A^*B^*$. It is the same as the minimal element in $h_C(\oi h_C(L))$ \wrt the ordering $a<b$.

\subsection{The power of $b$-torsion}\label{sec:rush}
Recall that for any subset $P\sse \Sig^*$ and $c\in \Sig^*$ we denote by $\pi_P(v)\in \N^P$ its Parikh image as defined in \prref{sec:parikh}. 
The following lemma shows a crucial ``downward-closure-property'' used in \prref{thm:allreg}.
\begin{lemma}\label{lem:fred}Let $L$, $C$, and $\rf$ as in \prref{def:nfwh} and in  \prref{def:AB}. 
  Let $v\in C^*$ and $w\in L$ such that $\pi_C(v)\leq \pi_C(\rf(w))$. If $\pi_{\os z}(v)\geq 1$ for all $z\in C$, then we have $\rho(v)\in \rho(L)$.  
\end{lemma} 
\begin{proof}
If $\pi_{\os z}(v)= \pi_{\os z}(\rf(w))$ for all $z\in C$, then 
$\rho(v)=\rho(w)$ and therefore $\rho(v)\in \rho(L)$ because $\rho(w)\in \rho(L)$ thanks to $w\in L$. Thus, we may assume that 
$1\leq \pi_{\os z}(v)< \pi_{\os z}(\rf(w))$ for some $z\in C$. 
Thus, without restriction we have  $\rf(w)=uzu'zu''$ with $z\in C$ and $u,u',u''\in C^*$ and $\pi_C(v)\leq \pi_C(\rf(w'))$ where 
$w'=uu'zu''$. We have $\rf(w')=w'$ by \prref{def:nfwh}. Moreover, a repetition of some $z\in C$ does not change the specified graph. 
Let  $\wh{w}'$ denote the saturation of $w'$. By \prref{lem:nfwh},
we have $\rho(\wh{w}')\sse \rho(L)$. Define $L'=L\cup \wh{w}'$. Then we have 
 $w'\in L'$ and $\rho(L)=\rho(L')$; moreover, $L'$ satisfies the same \btptorp as $L$ does. We can work with the same $C$, too.
Since $w'$ is shorter than $w$, we conclude that  by induction on the length of $w$ that $\rho(v)\in \rho(L')=\rho(L)$.  \end{proof} 
  \begin{theorem}\label{thm:allreg} 
  Let $L\sse \G$ be any language satisfying the \btorp. Then there is a regular 
  set  $R\sse \G$ such that  
   $\rho(L)= \rho(R)\,.$
 \end{theorem} 
  \begin{proof}The proof relies on Dickson's Lemma. We have $\rho(L)=\bigcup\set{\rho(\wh v)}{v\in L}$ by \prref{lem:nfwh}.  
  The set $\rho(\rf(L))$ is finite, as $L$ satisfies the \btorp. Thus, there is a finite 
  subset $K\sse L$ such that $\rho(\rf(L))=\bigcup\set{\rho(v)}{v\in K}$. Let $C\sse \E\cup \V$ be the finite subset such that $$C= \set{z\in \E\cup \V}{\exists v\in K: \pi_{\E\cup \V}(v)(z)\geq 1}\,.$$ 
  Splitting $L$ into disjoint subsets and replacing $C$ by various subsets of $C$ we may assume without restriction that for all $w\in L$, the corresponding $C$ contains (after the split) exactly all letters that are factors of $\rf(w)$ and moreover, for all $v\in C^*$, we have  
  $z\leq \rf(w)\iff \pi_{C}(w)(z)\geq 1$. 
  After this modification, there are no vectors in 
  $\pi_C(\rf(L))$ with zero-entries. The crucial observation is stated in  
  \prref{lem:fred}. The lemma tells us that we do not change $\rho(L)$ if we augment   
  $L$ by all words $v\in C^*$ where there is some $w\in L$  
  such that $1\leq \pi_C(v)(z) \leq \pi_C(\rf(w))(z)$ for all $z\in C$. 
  Therefore, we may assume without restriction that $\pi_C(\rf(L))$ 
  is positively downward-closed according to \prref{sec:parikh}.  
  We also explained in \prref{sec:parikh} that Dickson's Lemma implies that $\pi_C(\rf(L))$ is semi-linear.  
  By Parikh's Theorem, see \prref{prop:gs66parikh66}~(\ref{parikh66}), there is a regular set $R'\sse C^*$ of words such that 
  $\pi_C(R')= \pi_C(\rf(L))$.  
  The class of regular sets is closed under regular \subst{}s.  
  The inclusion $C\sse (\V \cup \E)^*$ defines a canonical \hom $h:C^*\to (\V \cup \E)^*$. 
  Hence, if we substitute in $R'$ every letter $z$ by $[h(z)]$, then   
  we obtain $\rho(R)=\rho(L)$, where  
  $R=\bigcup\set{\wh w}{w\in R'}$. 
 \end{proof}

 \begin{corollary}\label{cor:allreg}   
  Let $L\sse \G$ satisfy the \btorp. Then, given a finite graph $G=(V_G,E_G)$ as  an input,  
  it is decidable whether $G\in \rho(L)$.  
 \end{corollary} 
 The contents of \prref{cor:allreg} inspired the title of this subsection.
 \begin{proof} 
  By \prref{thm:allreg}, we may replace $L$ by some regular set $R$ where  $\rho(R)=\rho(L)$. 
  \Ip we can calculate the threshold $t$ and the period $p$ such that $R$ satisfies the \btptorp. 
  The set $\rho(\rf(R))$ is finite and effectively computable.  
  For every $F\in \rho(\rf(R))$, we compute the  
  short-lex normal-form $\gam(F)$ as defined above.  
  We obtain a finite set $W$ of words containing all those $\gam(F)$. Thus, the set $\wh W= \bigcup\set{\wh w}{w\in W}$ is effectively regular and it  holds that $\rho(\wh W)= \rho(R)$.  
  Let $m=\max\set{|aua|}{u\in \Sigma^*\wedge \exists w\in W: aua\leq w}$. Let $M=m|V_G||E_G|$. Hence,  $G\in \rho(R)$ 
  \IFF there is word $w\in \wh W$ of length at most~$M$ such that  
  $G=\rho(w)$.   
 \end{proof} 

\begin{corollary}\label{cor:Liscf}   
  Let $L\sse \G$ be context-free satisfying the \btptorp. Then, we can effectively calculate a regular set $R\sse \G$ such that  
  $\rho(R)=\rho(L)$. 
 \end{corollary} 
 \begin{proof} 
  Let $C=\set{z\in \V\cup \E}{\forall m\in \N: ab^ma \leq z \implies m\leq t+p-1}$. The   
  inclusion of $C\sse (\V\cup \E)^*$ defines  
  a \hom $h$ from the free monoid $C^*$ to $(\V\cup \E)^*$. 
  Hence, $L'= \oi h(L)$ is effectively context-free. Therefore,  
  $\pi_C(L')= \pi_C(\oi h(L))$ is effectively semi-linear.  
  The semi-linear set can be represented by a regular language $R'\sse C^*$.  
  As in the proof of \prref{thm:allreg}, we obtain $R$ as the image of $R'$ under the regular \subst which replaces every letter $z\in C$ by $[h(z)]$. 
 \end{proof} 

 Let $R\sse \G$ be regular. It is well-known that there might be a much more concise representation by some \cfree language $K\sse \G$ such that $\pi_C(K)=\pi_C(R)$ and hence $\rho(K)=\rho(R)$. Therefore, we might describe graph families even more concisely using context-free grammars than using NFAs.

\subsection{Switching the alphabet and Parikh images}\label{sec:newalp} 
 By \prref{thm:allreg}, we know that regular languages suffice to 
 describe all sets $\rho(L)$ where $L\sse \G$ satisfies the \btorp.  
 Therefore, we restrict ourselves to regular languages. Throughout this section, $R\sse \G$ denotes a regular language. Hence, we can calculate a threshold $t$ and a period $p\geq 1$ such that  
 $R$ satisfies the \btptorp. 
 Since $R$ is regular, the set $L= \oi h_C(R)\cap A^*B^*$ is regular, too;  its \emph{Parikh-image} $\pi_C(L)\sse \N^{C}$ is effectively semi-linear. (This is always true as soon as $R$ and hence $L$ is context-free.) 
 Hence, for some finite index sets $J$ and $I_j$ we can write  
 \begin{equation}\label{eq:piL} 
  \pi_C(L)=\bigcup_{j\in J}\left
(q_j + \sum_{i\in I_j}\N p_i\right
)\,, 
 \end{equation} 
 where $q_j, p_i \in \N^{C}$ are vectors. Splitting $\pi_C(L)$ into more linear sets by making the index set  
 $J$ larger and the sets $I_j$ smaller (if necessary), we can assume without restriction that, for all $j\in J$ and $z\in C$, we have  
 $\sum_{i\in I_j}p_i(z) \leq q_j(z)$. To see this, let $1\in I_j$. Then,  we have 
 \begin{align*} 
  q_j + \sum_{i\in I_j}\N p_i=  \left
(q_j + \sum_{i\in I_j\sm \os 1}\N p_i\right
) 
  \cup  \left
(q_j + p_1 + \sum_{i\in I_j}\N p_i\right
)\,. 
 \end{align*}

 Splitting $L$ into even more but finitely many cases, we can assume without restriction  
 (for simplifying the notation) that the set $J$ is a singleton.  
 Thus,  $\pi(L)=q + \sum_{i\in I}\N p_i$ for some  
 $q, p_i \in \N^{C}$ such that $\sum_{i\in I}p_i(z) \leq q(z)$ for all  $z\in C$.  
 Moreover, making $A,B,C$ perhaps smaller, we may assume without restriction that $q(z)\geq 1$ for all $z\in C$ and $C=A\cup B$.  
 (A similar argument was used in the proof of \prref{thm:allreg} above.) 
 
 In order to understand the set of graphs in $\rho(R)$, it suffices to understand the set of finite graphs defined by linear sets of the  
 form  $S=q + \sum_{i\in I}\N p_i\sse \N^C$, where $q(z)\geq 1$ for all $z\in C$ and $\sum_{i\in I}p_i\leq q$.  
 For that purpose, we let $r= \sum_{i\in I}p_i\leq q$ and we define a function $\alp:C\to \N_\infty$ as follows.  
 \begin{equation}\label{eq:alp} 
  \alp(z)= \begin{cases} 
   q(z)& \text{if } r(z) = 0 \wedge \exists m\in \N: t\leq m \wedge ab^ma\leq z\\  
   \infty & \text{if } r(z) \geq 1 \wedge \exists m\in \N: t\leq m \wedge ab^ma\leq z\\  
   1& \text{otherwise. That is: } \forall m\in \N: ab^ma\leq z \implies m<t.  \end{cases} 
 \end{equation} 
 For all $z\in C$, let $L_z\sse \Sigma^*$. Then, we introduce 
the following notation. 
\begin{align}\label{eq:Lz} 
\prod_{z\in C}L_z= L_{z_1} \cdots L_{z_{\abs {C}}} 
\end{align} 
 where $z_i\leq z_j$ for all $i\leq j$ according to the linear order defined in \prref{eq:linC}. Observe that $\prod_{z\in C}L_z$ 
 is regular if all $L_z$ are regular.  
 With this notation, we define  regular sets $R_\alp \sse \Sigma^*$ and $L_\alp\sse \Sigma^*$ by  
 \begin{align}\label{eq:RLalp} 
  R_\alp=\prod_{z\in C}z^{\alp(z)} \quad \text{and} \quad  L_\alp=\prod_{z\in C}[z]^{\alp(z)} 
 \end{align} 
Here and in the following, $L^\infty$ is just another notation for $L^+$ if $L$ is any set of words.
%
\begin{theorem}\label{thm:RL}The sets $R_{\alp}$ and $L_{\alp}$ are regular sets with $R_{\alp}\sse L_{\alp}$  
  and $\rho(L_{\alp}) = \rho(R)$. 
 \end{theorem} 
\begin{proof} 
Without restriction, \prref{eq:piL} reads as  
  $\pi_C(L) = q + \sum_{i\in I} \N p_i$. 
  As $z\in [z]$,  the inclusions $R_{\alp}\sse L_{\alp}$ and  
  $\rho(R_{\alp}) \sse \rho(R)$ follow by definition.  
  For the converse, let $v\in R$ and $G=\rho(v)$. Choose some $w\in L_\alp$ with $\pi_C(w) = q + \sum_{i\in I} m_i p_i$. Choosing $m\in \N$ large enough, we find  
   $\pi_C(v) \leq q + mr \in \pi_C(\rf(L_{\alp}))\,$ where, as above, $r=\sum_{i\in I} p_i$.
  Hence, we can apply \prref{lem:fred} to finish the argument.
 \end{proof} 

\subsection{Marked graphs}\label{sec:marked} 
 In the next steps, we define for $\alp$ a finite family  
 of finite graphs $\cF_\bet$, together with a marking on vertices and edges. Each member of $\cF_\bet$ is just a singleton of a marked graph. The family $\cF_\bet$  still depends on $\alp$, only, but it might be more indices $\bet$ than indices $\alp$ for defining $\cF_\alp$. 
Then, for each $F\in \cF_\bet$, we define a possibly infinite graph $F^\infty$, using the notion of a marked graph.  \begin{definition}\label{def:Falp} 
  For $z\in C$ let $\alp'(z)= \alp(z)$ if $\alp(z)<\infty$ and $\alp'(z)= 1$, otherwise.  
  We let $w_{\alp}=\prod_{z\in C}z^{\alp'(z)}$, and we define $\cF_\alp$ by $\cF_\alp=\rho(R_{\alp})$ where $R_\alp$ is as in \eqref{eq:RLalp}.  
 \end{definition} 

 Since $w_\alp$ is a single word, the set $\cF_\alp$ is  a finite set of finite graphs.  
 We are now defining the crucial notion of a marked graph, 
where some vertices and edges are marked.

\begin{definition}\label{def:mark} 
A \emph{marked graph} is a tuple $F=(V_F,E_F,\mu)$, where $(V_F,E_F)$ is a finite graph and $\mu\sse V_F\cup E_F$ denotes the set of marked vertices and edges.
  Isolated vertices may appear, but if an isolated vertex is marked, then there is exactly one isolated vertex.  
  We also require that whenever an edge $(u,v)$ is marked, then at least one of its endpoints is marked, too. 
  A \emph{marked edge-graph} is a marked graph without isolated vertices.  
 \end{definition} 
In a first step we let the word $w_\alp$ and the set $R_\alp$  as defined above in \prref{def:Falp}. Performing another splitting, we can write $R_\alp$ as a finite union of sets 
$R_{\alp,k}$  such that each $R_{\alp,k}$ defines exactly one graph $(V_{\alp,k},E_{\alp,k})$. By making, if necessary, 
the alphabet $C$ as well as the threshold $t$ larger (but keeping the period $p$), we assume that for each $(\alp,k)$, there is some $\bet:C\to \N_\infty$ such that 
\begin{equation}\label{eq:bet}
\rho(\prod_{z\in C}z^{\bet(z)})= (V_{\alp,k},E_{\alp,k})
\end{equation}
where $\bet(z)\in \os{0,1,\infty}$ with the requirements that first, 
$\bet(z)=\infty \implies [z]\neq \os z$ and second, 
if $z=uav$ denotes an edge with $[z]\neq \os z$, then 
$\bet(z)=\infty \iff ([u]\neq \os {u}\vee [v]\neq \os {v})$.
In order to simplify the notation, we concentrate on one 
$(\alp,k)$ and we denote $(V_{\alp,k},E_{\alp,k})$ as $(V_{F'},E_{F'})$.

We are ready to define the marking for $(V_{F'},E_{F'})$.
We mark exactly those vertices and edges $z$ in $V_{F'}\cup E_{F'}$ where $\bet(z)= \infty$. Note that therefore an edge is marked only if at least one of its endpoints is marked. It might be, however, that there is an unmarked edge $(u,v)$ where $u$ and $v$ are marked. 
In the final step, if an isolated vertex is marked, then remove all isolated marked vertices  
 $y$  except one isolated vertex $y$ which is marked. \Ip after that procedure if a marked isolated vertex $y$ appears, then $\bet(y)= \infty$. This reduces the number of marked graphs without changing 
 the set $\rho(R)$ we were interested in.
 
We denote by $\cF_\bet$ this family of marked graphs which was defined through the function $\alp$.

Let us switch to a more abstract viewpoint.  
 We let $\cF$ be any finite family of marked graphs. For each  
 $F=(V_F,E_F,\mu)\in \cF$, we define a possibly infinite graph $F^\infty$  
 where $(V_F,E_F)$ appears as an induced subgraph, and we define a  
 family $\cG_F$ of finite graphs. 
  In the application, we consider finitely many $\cF_\bet$,  and
 then we study $\bigcup\set{\cG_F}{F\in \cF_\bet}$, where $F=(V_F,E_F,\mu)$ is the marked graph obtained by the canonical marking procedure above (which might have removed isolated vertices).  
 It turns out that, for a full description of $\rho(R)$, it is enough to describe sets  
 $\cG_F$ for marked graphs $F=(V_F,E_F,\mu)$. This requires to define~$F^\infty$. 

 \begin{definition}\label{def:Falpinf} 
  Let $F=(V_F,E_F,\mu)$ be a marked graph as in  \prref{def:mark}. 
  Then, the  graph $F^\infty=(V_F^\infty,E_F^\infty)$  
  is defined as follows. 
  \begin{align*} 
   V^\infty_F&=V_F\times \os 0 \cup\bigcup_{u\in V_F}\set{(u,k)}{u \text{ is marked } \wedge k\in \N}\,,\\ 
   E^\infty_F&= E_F\times \os 0 \cup \{((u,k),(v,\ell))\in V_F^\infty\times V_F^\infty\mid (u,v)\in E_F\wedge (u,v) \text{ is marked}\} \,,
  \end{align*} 
with $E_F\times \os 0=\{((u,0),(v,0))\mid (u,v)\in E_F\}$. 
  The family $\cG_F$ is the set of finite subgraphs of $F^\infty$ 
   containing $(V_F\times \os 0,E_F\times \os 0)$ as an induced subgraph.    
 \end{definition} 
 Observe that $F^\infty=(V_F,E_F)$ \IFF there is no marking, \ie, if $\mu=\es$.  
 We embed $F$ into $F^\infty$ by a graph \morph $\gam$  
 which maps each vertex $u\in V_F$ to the pair $\gam(u) = (u,0)\in V^\infty_F$. The projection onto the first component $\phi(u,k)=u$ yields a retraction for every $G\in \cG_F$ with retract $F$. If no isolated vertex is marked, then  $F^\infty$ has at most $|V_F|$ isolated vertices, but if there are marked vertices, then for every sufficiently large $k$, there is some graph in  
 $\cG_F$ which has exactly $k$ isolated vertices. In order to understand the graphs in $\cG_F$ (which is our goal), it is enough to understand the graphs~$G$ satisfying  
 $F\leq G\leq F^\infty$. For $F=F^\infty$, we have the full information about that set. Thus, we focus on $F\neq F^\infty$. 
 \prref{thm:Fiffi} shows that   
 $\rho(R)$ is rather rich as soon as some  
 $F\in \cF_\bet$ satisfies $F\neq F^\infty$. 
\begin{figure}[h]\centering
\tikzset{every loop/.style={}}
\begin{tabular}{ccccc}
	\begin{tikzpicture}
		\tikzset{mystyle/.style={state, inner sep=2pt,minimum 	size=1pt, fill=black, node distance=20pt}}
\tikzset{markedstyle/.style={state, inner sep=.5pt,minimum 	size=1pt, fill=white, node distance=20pt}}
		\node[mystyle] (0) {};
		\node[mystyle, above of = 0] (1) {};
		\node[mystyle, above of = 1] (2) {};
		\node[mystyle, above of = 2] (3) {};
		\node[above of = 3] (dots) {$\vdots$};
\node[markedstyle,below =of 0](m){$\ast$};
	\end{tikzpicture}
&
\begin{tikzpicture}
	\tikzset{mystyle/.style={state, inner sep=2pt,minimum 	size=1pt, fill=black, node distance=20pt}}
\tikzset{markedstyle/.style={state, inner sep=.5pt,minimum 	size=1pt, fill=white, node distance=20pt}}
	\node[mystyle] (0) {};
	\node[mystyle, above of = 0] (1) {};
	\node[mystyle, above of = 1] (2) {};
	\node[mystyle, above of = 2] (3) {};
	\node[above of = 3] (dots) {$\vdots$};
	\node[mystyle] (c) at (1, 0) {};
\node[markedstyle,below =of 0](m0){$\ast$};
\node[mystyle,below =of c](m1){};
	\path[line width=3pt] (m0) edge[draw=black!40] (m1);
	\path 
	(0) edge (c)
	(1) edge (c)
	(2) edge (c)
	(3) edge (c)
	(dots) edge[dotted] (c);
\end{tikzpicture}
&
\begin{tikzpicture}
	\tikzset{mystyle/.style={state, inner sep=2pt,minimum 	size=1pt, fill=black, node distance=20pt}}
\tikzset{markedstyle/.style={state, inner sep=.5pt,minimum 	size=1pt, fill=white, node distance=20pt}}
	\node[mystyle] (0) {};
	\node[mystyle, above of = 0] (1) {};
	\node[mystyle, above of = 1] (2) {};
	\node[mystyle, above of = 2] (3) {};
	\node[above of = 3] (dots) {$\vdots$};
	\node[mystyle] (r0) at (1,0){};
	\node[mystyle, above of = r0] (r1) {};
	\node[mystyle, above of = r1] (r2) {};
	\node[mystyle, above of = r2] (r3) {};
	\node[above of = r3] (rdots) {$\vdots$};
\node[markedstyle,below =of 0](m0){$\ast$}; 
\node[markedstyle,below =of c](m1){$\ast$};
	\path[line width=3pt] (m0) edge[draw=black!40] (m1);
	\path 
	(0) edge (r0)
	(0) edge (r1)
	(0) edge (r2)
	(0) edge (r3)
	(0) edge[dotted] (rdots)
	(1) edge (r0)
	(1) edge (r1)
	(1) edge (r2)
	(1) edge (r3)
	(1) edge[dotted] (rdots)
	(2) edge (r0)
	(2) edge (r1)
	(2) edge (r2)
	(2) edge (r3)
	(2) edge[dotted] (rdots)
	(3) edge (r0)
	(3) edge (r1)
	(3) edge (r2)
	(3) edge (r3)
	(3) edge[dotted] (rdots)
	(dots) edge[dotted] (r0)
	(dots) edge[dotted] (r1)
	(dots) edge[dotted] (r2)
	(dots) edge[dotted] (r3)
	(dots) edge[dotted] (rdots);
\end{tikzpicture}
&
\begin{tikzpicture}[every loop/.style={}]
	\tikzset{mystyle/.style={state, inner sep=2pt,minimum 	size=1pt, fill=black, node distance=20pt}}
\tikzset{markedstyle/.style={state, inner sep=.5pt,minimum 	size=1pt, fill=white, node distance=20pt}}
	\node[mystyle] (0) at (4.5,0){}        edge [out=118,in=62,draw=black,loop]();
	\node[mystyle, above of = 0] (1) {} edge [out=118,in=62,draw=black,loop]();
	\node[mystyle, above of = 1] (2) {} edge [out=118,in=62,draw=black,loop]();
	\node[mystyle, above of = 2] (3) {} edge [out=118,in=62,draw=black,loop]();
	\node[above of = 3] (dots) {$\vdots$};
\node[markedstyle,below =of 0](m0){$\ast$} edge [out=140,in=40,draw=black!40,line width=3pt,loop] ();
\path
	(0) edge (1)
	(0) edge[bend left] (2)
	(0) edge[bend left] (3)
	(0) edge[bend left, dotted] (dots)
	(1) edge (2)
	(1) edge[bend left] (3)
	(1) edge[bend left, dotted] (dots)
	(2) edge (3)
	(2) edge[bend left, dotted] (dots)
	(3) edge[bend left, dotted] (dots)
;
\end{tikzpicture}
&
\begin{tikzpicture}[every loop/.style={}]
	\tikzset{mystyle/.style={state, inner sep=2pt,minimum 	size=1pt, fill=black, node distance=20pt}}	
\tikzset{markedstyle/.style={state, inner sep=.5pt,minimum 	size=1pt, fill=white, node distance=20pt}}
	\node[mystyle] (f0) at (2,0) {};
	\node[mystyle] (l0) at (3,0) {};
	\node[mystyle, above of = l0] (l1) {};
	\node[mystyle, above of = l1] (l2) {};
	\node[mystyle, above of = l2] (l3) {};
	\node[above of = l3] (ldots) {$\vdots$};
	\node[mystyle] (0) at (4.5,0){}        edge [out=118,in=62,draw=black,loop]();
	\node[mystyle, above of = 0] (1) {} edge [out=118,in=62,draw=black,loop]();
	\node[mystyle, above of = 1] (2) {} edge [out=118,in=62,draw=black,loop]();
	\node[mystyle, above of = 2] (3) {} edge [out=118,in=62,draw=black,loop]();
	\node[above of = 3] (dots) {$\vdots$};
	\node[mystyle] (r0) at (6,0){} edge [out=140,in=40,draw=black,loop]();
	\node[mystyle, above of = r0] (r1) {} edge [out=118,in=62,draw=black,loop]();
	\node[mystyle, above of = r1] (r2) {} edge [out=118,in=62,draw=black,loop]();
	\node[mystyle, above of = r2] (r3) {} edge [out=118,in=62,draw=black,loop]();
	\node[above of = r3] (rdots) {$\vdots$};
\node[mystyle] (s0) at (7.5,0){};
\node[mystyle] (t0) at (8.5,0){}  edge [out=140,in=40,draw=black,loop]();
\node[below of =f0](f0b){};
\node[mystyle,below =of f0](m0){};
\node[markedstyle,below =of l0](m1){$\ast$};
\node[markedstyle,below =of 0](m2){$\ast$} edge [out=140,in=40,draw=black!40,line width=3pt,loop]();
\node[markedstyle,below =of r0](m3){$\ast$} edge [out=140,in=40,draw=black!40,line width=3pt,loop]();\node[mystyle,below =of s0](m4){};
\node[mystyle,below =of t0](m5){} edge [out=140,in=40,draw=black,loop]();
\path 
(m0) edge [draw=black!40,line width=3pt] (m1)
(m1) edge [draw=black!40,line width=3pt] (m2)
(m2) edge [draw=black!40,line width=3pt] (m3) 
(m3) edge [draw=black!40,line width=3pt] (m4)
(m4) edge (m5)
;
	\path 
	(f0) edge (l0)
	(f0) edge (l1)
	(f0) edge (l2)
	(f0) edge (l3)
	(f0) edge[dotted] (ldots)
	(l0) edge (0)
	(l0) edge (1)
	(l0) edge (2)
	(l0) edge (3)
    (l0) edge[dotted] (dots)
	(l1) edge (0)
	(l1) edge (1)
	(l1) edge (2)
	(l1) edge (3)
	(l1) edge[dotted] (dots)
	(l2) edge (0)
	(l2) edge (1)
	(l2) edge (2)
	(l2) edge (3)
	(l2) edge[dotted] (dots)
	(l3) edge (0)
	(l3) edge (1)
	(l3) edge (2)
	(l3) edge (3)
	(l3) edge[dotted] (dots)
	(ldots) edge[dotted] (0)
	(ldots) edge[dotted] (1)
	(ldots) edge[dotted] (2)
	(ldots) edge[dotted] (3)
	(ldots) edge[dotted] (dots)
	(0) edge (r0)
	(0) edge (r1)
	(0) edge (r2)
	(0) edge (r3)
	(0) edge[dotted] (rdots)
	(0) edge (1)
	(0) edge[bend left] (2)
	(0) edge[bend left] (3)
	(0) edge[bend left, dotted] (dots)
	(r0) edge (r1)
	(r0) edge[bend right] (r2)
	(r0) edge[bend right] (r3)
	(r0) edge[bend right, dotted] (rdots)
	(1) edge (r0)
	(1) edge (r1)
	(1) edge (r2)
	(1) edge (r3)
	(1) edge[dotted] (rdots)
	(1) edge (2)
	(1) edge[bend left] (3)
	(1) edge[bend left, dotted] (dots)
	(r1) edge (r2)
	(r1) edge[bend right] (r3)
	(r1) edge[bend right, dotted] (rdots)
	(2) edge (r0)
	(2) edge (r1)
	(2) edge (r2)
	(2) edge (r3)
	(2) edge[dotted] (rdots)
	(2) edge (3)
	(2) edge[bend left, dotted] (dots)
	(r2) edge (r3)
	(r2) edge[bend right, dotted] (rdots)
	(3) edge (r0)
	(3) edge (r1)
	(3) edge (r2)
	(3) edge (r3)
	(3) edge[dotted] (rdots)
	(3) edge[bend left, dotted] (dots)
	(r3) edge[bend right, dotted] (rdots)
	(dots) edge[dotted] (r0)
	(dots) edge[dotted] (r1)
	(dots) edge[dotted] (r2)
	(dots) edge[dotted] (r3)
	(dots) edge[dotted] (rdots)
(s0) edge (r0) 
(s0) edge (r1)
(s0) edge (r2)
(s0) edge (r3)
(s0) edge[dotted] (rdots)
(s0) edge (t0)
;
\end{tikzpicture}
\\
(a) & (b) & (c) & (d) &(e)
%
%
\end{tabular}
\caption{The four basic situations how a marked graphs $F=(V_F,E_F)$ can lead to a graph $F^\infty$ are shown in (a) through (d). The right-most picture (e) combines  various situations. The marked graphs~$F$ are shown in the lower line, and above each marked graph 
the graphs $F^\infty$ are shown in black. The lowest line in each graph $F^\infty$ shows the embedding of $(V_F,E_F)$ into $F^\infty$.}\label{fig:markedgraphs}
\end{figure}
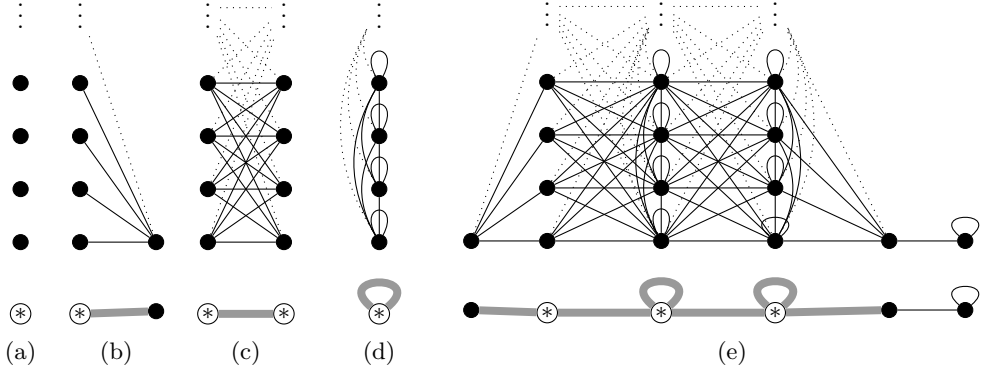
The reader might find it helpful to look at the examples of Figure~\ref{fig:markedgraphs} to understand the building of the graph $F^\infty$ from a marked graph $F$ if at least one vertex is marked. The pictures show five different situations: (a) a graph $F$ with a single isolated marked vertex: it gives rise to an arbitrary number of isolated vertices in $F^\infty$; (b) a graph $F$ with a single marked vertex incident to a marked edge yields an infinite star in $F^\infty$; (c) a graph $F$ with a marked edge with different endpoints where both are marked. It  yields an infinite complete bipartite graph in $F^\infty$; 
(d) a graph $F$ with a marked self-loop gives rise to an infinite clique in $F^\infty$; (e)
The rightmost figure combines the last three situations in a single marked graph, plus the effect of an unmarked edge.

\begin{theorem}\label{thm:Fiffi} 
  Let $F=(V_F,E_F,\mu)$ be any marked graph. 
\begin{enumerate}
\item If there is no marking, then $\cG_F=\os{(V_F,E_F)}$.
\item If $F$ contains a marked edge $(u,v)$ where the endpoint $v$ is marked, then every finite star with center $(u,0)$ appears as an induced subgraph of some $G\in \cG_F$. 
\item Suppose we represent a bipartite graph as a triple $(U,V,E)$  
where $U\cap V=\es$ and $E\sse U\times V$. Let $H$ be any finite bipartite edge-graph. If $F$ contains a marked edge $(u,v)$ where $u$ and~$v$ are marked, then a disjoint union of $F$ and $H$ appears in~$\cG_F$. 
\item Let $H$ be any finite graph. If $F$ contains a marked self-loop $(u,u)$, then the disjoint union of $F$ and $H$ belongs to $\cG_F$. 
   \item Let $F$ be any marked graph such that at most two vertices are marked. Then the following  holds. 
   A disjoint union of $F$ and a triangle (or any other non-bipartite graph) appears in $\cG_F$ \IFF there is some marked self-loop in~$F$.
\end{enumerate} 
\end{theorem}
 \begin{proof}Define the vertex sets $V_0=V_F\times\os 0$ and $V_{\geq 1}=\{(u,k)\mid u\in V_F\land k\geq 1\}$. 
We consider the five cases separately.
\begin{description}
\item[{1.}] By definition.
\item[{2.}]
In the graph $F^\infty$, there are directed edges from 
 $(u,0)$ to all $(v,k)$, where $k\in \N$. \Ip every finite star with center $(u,0)$ appears as an induced subgraph of some $G\in \cG_F$.
\item[{3. and 4.}] 
Consider a marked edge $(u,v)$ with both endpoints marked. 
  In the graph $F^\infty$, there are directed edges from 
 $(u,k)$ to all $(v,\ell)$, where $k,\ell\in \N$. Consider the induced subgraph $F^\infty[U]$ where $U= V_{\geq 1}$.  By definition, $F^\infty[U]$ is disjoint from the subgraph $F=F^\infty[V_0]$. 
For $u=v$, the graph $F^\infty[U]$ is an infinite complete graph; for $u\neq v$, the graph $F^\infty[U]$ is an infinite complete bipartite graph. The claims follow. 
\item[{5.}]
If $F$ contains a marked self-loop, then, by definition, it is a self-loop around a marked vertex. Hence, we are done since 
every disjoint union of $F$ and any other finite graph $G$ appears in $\cG_F$. For the other direction, assume that $F$ has no marked self-loop. If a disjoint union of $F$ and a finite non-bipartite graph $G$ appears in $\cG_F$, then we need at least three marked vertices to produce~$G$. \qedhere
\end{description}
\end{proof}  
The following lemma uses the notions of vertex cover and of bag-size. Recall that the tree-width was actually defined by bag-size - 1 by Robertson and Seymour in \cite{RobSey84}.
\begin{lemma}\label{lem:btw}
Let $F=(V_F,E_F,\mu)$ be a nonempty marked graph such that 
each marked edge has at least one unmarked vertex incident to it. Then for every $G\in\cG_F$ both, the size of a minimal vertex cover and its bag-size are bounded by $|V_F|$.
\end{lemma}
The following proof is based on well-known and standard techniques.
\begin{proof}
Let $G=(V_G,E_G)\in\cG_F$. Since every marked edge has at most one marked vertex incident to it, $V_G$ is the disjoint union of
$V_F$ and an independent set $U$. Clearly, $V_F$ is a vertex cover of $G$. 
We now construct a tree-decomposition as follows. We begin with 
single bag $B$ defined by the vertex cover $V_F$. Then for every $u\in U$ we define a bag $B_u$ by $N(u)\cup \os u$ where $N(u)$ is the set of neighbors of $u$. Note that $N(u)\sse V_F$. A bag $B_u$ connected to $B$ \IFF $N(u)\neq \es$. 
This is tree decomposition of $G$ where the connected component of $G$ (in the underlying tree) is a star with $|U|$ rays.
\end{proof}
For an illustration of our construction, we refer to \prref{fig:starandclique}.
\begin{remark}\label{rem:baum}
The construction in the proof of \prref{lem:btw} is optimal with respect to the minimal vertex cover and to the bag-size if $(V_F,E_F)$ is a clique  and 
$S=(V_F,S_F)$ is its subgraph of marked edges such that $S$ is a star, where the center of the star is not marked. In general, 
we might achieve smaller vertex covers and bag-sizes by beginning  
with a tree decomposition of $F$. 
\end{remark}
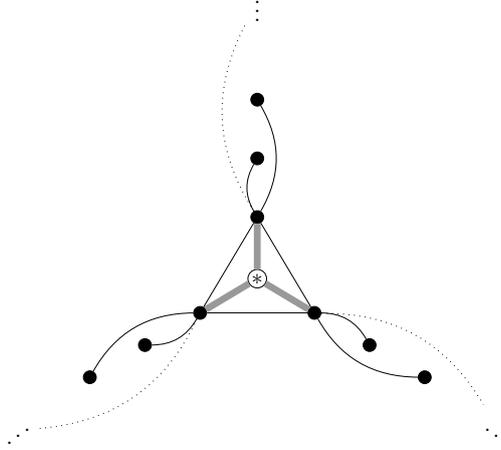
\begin{figure}[h]
\vspace{-0.8cm} 
\centering
\scalebox{.85}{\begin{minipage}{.6\textwidth} 
\tikzset{every loop/.style={}}
	\begin{tikzpicture}
		\tikzset{mystyle/.style={state, inner sep=2pt,minimum 	size=1pt, fill=black, node distance=20pt}}
\tikzset{markedstyle/.style={state, inner sep=.5pt,minimum 	size=1pt, fill=white, node distance=20pt}}
		\node[markedstyle] (0) {$\ast$};
		\node[mystyle, above = of 0] (u1) {};
		\node[mystyle, above = of u1] (u2) {};
		\node[mystyle, above = of u2] (u3) {};
		\node[above = of u3] (udots) {$\vdots$};
		\node[mystyle, below left = 1em and 2em of 0] (l1) {};
		\node[mystyle, below left = 1em and 2em of l1] (l2) {};
		\node[mystyle, below left = 1em and 2em of l2] (l3) {};
		\node[below left = 1em and 2em of l3] (ldots) {\reflectbox{$\ddots$}};
		\node[mystyle,below right  = 1em and 2em of 0] (r1) {};
		\node[mystyle,below right  = 1em and 2em of r1] (r2) {};
		\node[mystyle, below right = 1em and 2em of r2] (r3) {};
		\node[below right  = 1em and 2em of r3] (rdots) {$\ddots$};
\path [line width=3pt] (0) edge[draw=black!40] (u1) 
(0) edge[draw=black!40] (l1) 
(0) edge[draw=black!40] (r1);
\path (u1) edge (l1) (l1) edge (r1) (r1) edge (u1) 
(u1) edge[bend left] (u2)
	(u1) edge[bend right] (u3)
	(u1) edge[bend left, dotted] (udots)
(l1) edge[bend left] (l2)
	(l1) edge[bend right] (l3)
	(l1) edge[bend left, dotted] (ldots)
(r1) edge[bend left] (r2)
	(r1) edge[bend right] (r3)
	(r1) edge[bend left, dotted] (rdots)
;
	\end{tikzpicture}
\end{minipage}}
\caption{An illustration of \prref{lem:btw} with $n=4$. The graph $F$ is a $K_4$, but the marked edges induce a star with three rays. 
In $F^\infty$, the star sends three streams of rays into infinity.}\label{fig:starandclique}
\end{figure}

\begin{definition}\label{def:4cls}
For $i=1\lds 4$, we define
classes $\cC_i$ containing the sets $\rho(L)$ 
where $L\sse\G$ is regular. If $\cG$ denotes such a set $\rho(L)$, then:
\begin{enumerate}
\item We let $\cG\in \cC_1$ if $\cG$ is a finite set of graphs.
\item We let $\cG\in \cC_2$ if the set $\cG$ has bounded tree-width.
\item We let $\cG\in \cC_3$ if $\cG\in \cC_2$ or if there exists a finite set of graphs $\cF'$ 
such that for every finite bipartite graph $G'$ there is some 
$F'\in \cF'$ such that the disjoint union of $F'$ and $G'$ appears in $\cG$. 
\item We let $\cG\in \cC_4$ if $\cG\in \cC_3$ if there exists a finite set of graphs $\cF'$ 
such that for every finite graph $G'$ there is some 
$F'\in \cF'$ such that the disjoint union of $F'$ and $G'$ appears in $\cG$. 
\end{enumerate}
\end{definition}
Note that our definition enforces $\cC_i\sse \cC_j$ for $1\leq i \leq j \leq 4$. The reason is driven by our motivation to find
graphs in sets $\rho(L)$ satisfying some property given $\Phi$. 
For that, let us consider a marked graph $F=(V_F,E_F,\mu)$. 
Recall that $\mu$ is a subset of $V_F\cup E_F$. Now, if $\mu'\sse \mu$, then 
$\cG_{F'}\sse \cG_{F}$ where $F'=(V_F,E_F,\mu')$. Thus, increasing the marking increases  $\cG_{F}$, and this makes it more likely to find a graph satisfying $\Phi$. If there is no marking at all, then $\cG_{F}$ is finite, hence of bounded tree-width. As long as $\mu$ is without any marked edge, we remain in the class $\cC_2$. Suppose that 
$\mu$ contains a marked edge which is not a self-loop, then we mark first one of its endpoints but not the other one. We remain in $\cC_2$. 
If according to $\mu$ both endpoints are marked, we mark the second endpoint, too. The result is that we are now in $\cC_3$. In the final step we mark all self-loops. If there is at least one, we are in the class~$\cC_4$. Thus, step by step, starting with $(V_F,E_F,\es)$ we can make the intermediate families $\cG_{F'}$ larger; and we end in the largest family $\cG_F$. 

\begin{corollary}\label{cor:4cls}
Let $F=(V_F,E_F,\mu)$ is a marked graph. Then the following holds.
\begin{enumerate}
\item We have $\rho(\cG_F)\in \cC_1$ \IFF there is no marked vertex.
\item If there is no marked edge $(u,v)$ where both $u$ and $v$ are marked, then $\rho(\cG_F)\in \cC_2$. This implies that $\rho(\cG_F)$ has bounded tree-width. The tree-width and the the minimal size of a vertex cover are bounded by $|V_F|$. 
\item If there is a  marked edge $(u,v)$ where $u$ and $v$ are marked, then $\rho(\cG_F)\in \cC_3$.
This implies that every connected finite bipartite graph appears as a connected component of some $G\in \rho(\cG_F)$. \Ip for every $k\in \N$ there some $G\in \cG_F$ where its tree-width and the minimal size of a vertex cover are both greater than $k$.
\item We have $\rho(\cG_F) \in\cC_4$ \IFF every connected finite graph appears as a connected component of some $G\in\rho(L)$.
\end{enumerate}
\end{corollary}

\begin{proof}
It is enough to prove the lemma when $\rho(L) = \cG_F$ where $F=(V_F,E_F,\mu)$ is a marked graph. The result is now a direct consequence of \prref{thm:Fiffi} and \prref{lem:btw} for the second item concerning the bounded tree-width. 
\end{proof}
We now show in \prref{cor:aper} that for star-free languages, the classification only contains three cases. As an example, consider $R=(ab^+aaab^+a)^{+}$; then $R$ is star-free\footnote{The minimal DFA accepting $R$ is easily seen to be \emph{counter-free} in the sense of \cite{mp71}. Alternatively, one can check with a FO-sentence that a word~$w$ contains the factor $aaa$, but no factor $bab$. Every occurrence of the factor $aaa$ of $w$ is preceded by $ab^+$ and succeeded by $b^+a$. Moreover, the word~$w$ has a prefix in $ab^+aaa$ and a suffix in $aaab^+a$                      
and if there is a factor $baab$ in~$w$, then this factor is (immediately) preceded and followed by the factor $aaa$.} 
and  $\rho(R)\in\cC_4$, because $\rho(R)=\cG_F$, where $F$ is a marked self-loop around a marked vertex. Hence, $\rho(R)\in\cC_4$. Therefore, if $\rho(L)$ equals $\rho(R)$ for a star-free regular language $R$, then \prref{cor:aper} states that $\rho(L)\in \cC_3$ implies $\rho(L)\in \cC_4$. Hence, if $R$ is star-free, then $\rho(R)$ belongs to three classes, only.\footnote{\prref{cor:aper} corrects a misprint in \cite[Cor.~3]{DiekertFW_DLT21}.}  
\begin{corollary}\label{cor:aper} 
  Let $R$ be a star-free language such that $\rho(R)$ is infinite. Then, we have either $\rho(R)\in \cC_2$ or $\rho(R)\in \cC_4$.
 \end{corollary}
 \begin{proof} 
 Recall Sch\"utzenberger's classical theorem that a language $R$ is star-free \IFF first, it is regular and second, its syntactic 
 monoid $M_R$ is aperiodic, see~\cite{sch65sf}. Therefore, $R$ satisfies the \btptorp for some threshold $t\geq 0$ with period $p=1$.  We may assume that $\rho(R)$ is infinite. Then, 
 we find a marked graph $F$ such that 
 $\cG_F\sse\rho(R)$ and $\cG_F$ is infinite.
\footnote{The construction of the marked graphs changed perhaps the threshold, but not the period which is therefore still $1$. For simplicity, the new threshold is still called $t$.}
Hence, $F$ contains a marked edge where at least one endpoint is marked. This marked vertex is defined by the word $ab^{t+1}a$. Thus, it is unique because every vertex in $F$ has the form $ab^{i}a$ for 
 $1\leq {t+1}$ and the vertices $aba,\ldots, ab^t$ are not marked. 
 Since $\cG_F\notin \cC_2$, the marked graph $F$ contains a marked self-loop $ab^{t+1}aaab^{t+1}a$.
We now apply \prref{thm:Fiffi} and \prref{cor:4cls}. A marked self-loop implies 
$\rho(R)\in \cC_4$.
\end{proof} 
\section{Graph properties}\label{sec:GP} 
A \emph{graph property} is a \textbf{decidable} subset $\Phi\sse \G$. For a finite graph~$G$, we write  
 $G\models \Phi$ if the short-lex representation $\gam(G)$ belongs to $\Phi$.\footnote{Recall that  
 $\rho(w)$ is realized as a graph with a natural linear order on the vertices: we have $ab^ca\leq ab^da \iff c\leq d$.} 
 Given a word $w\in \G$, we can compute $\gam\rho(w)$. Hence, we can assume without restriction that $\Phi$ is saturated: $\oi\rho(\rho(\Phi))= \Phi.$ 
To simplify our presentation, we focus on properties of undirected finite graphs (without self-loops).  
This can be achieved by making the set $\Phi$ larger 
such that $\Phi$ has the following desired property:
If $u\in \G$ represents the  graph $\rho(u)= (V,E)$ and $\Phi$ speaks about undirected  graphs (resp.~undirected  graphs without self-loops) then $\rho(u)\models \Phi\iff (V,E\cup \oi E)\models \Phi$
(resp.~$\rho(u)\models \Phi\iff (V,(E\cup \oi E)\sm \id_V)\models \Phi$). 
  
We are interested in the following (uniform) satisfiability problem $\sat(\cG_F,\Phi)$.
 \begin{itemize} 
  \item Input: A marked graph $F$ and a graph property $\Phi\sse \G$. 
  \item Question: ``$\exists G\in \cG_F: G\models \Phi$?''  
 \end{itemize} 
Throughout this section, $F$ denotes a marked graph and  
 $\cG_F$ denotes the family of graphs defined in \prref{def:Falpinf}. 
Sometimes, it will be crucial in the following that $F$ or $\Phi$ are fixed. We will clarify this by writing $\sat_F(\cdot)$ or $\sat_\Phi(\cdot)$, respectively.
 
 For various well-studied graph properties, the satisfiability problem is always decidable.  This includes problems where  
 $\Phi$ states that a graph is planar (resp.~is closed under graph-minors, resp.~perfect, $k$-colorable, etc.). This is a direct consequence of the following fact. 
\begin{theorem}\label{thm:trivia}   
  Let either $\cG_F$ be finite (\ie, a singleton) or $\Phi$ be any graph property which is closed under taking induced subgraphs (or both).  
  Then, $\sat(\cG_F,\Phi)$ is decidable. 
 \end{theorem} 
 \begin{proof} 
  Since $F$ is an induced subgraph for every $G\in \cG_F$, it is enough to check whether $F\models \Phi$. This is possible, because $\Phi$ is decidable by definition.  
 \end{proof} 
 In many cases, graph properties are expressible either in monadic second-order logic (MSO for short) or in first-order logic (FO for short). MSO is a rich and versatile class to define graph properties. Moreover, we allow quantification over both, sets of vertices and sets of edges. Since $w\in \G$ defines graphs with a linear order, we can express in MSO, for example,  that the  
 number of vertices is even.  
 We use the following well-known results 
 as a black box. 
 First, given an MSO-sentence $\Phi$ and $k\in \N$, it is decidable whether  
 there exists a graph of tree-width at most $k$  satisfying~$\Phi$, see, \eg, \cite{CourcelleHandbookGG97,CourcelleEngelfriet2012,Seese91}. 
As a second black box, we use Trakhtenbrot's Theorem~\cite{Tra50}: on input of an FO-sentence~$\Phi$, it is 
undecidable whether  there exists a finite graph satisfying~$\Phi$.
\begin{remark}\label{rem:tra50}
Trakhtenbrot's Theorem also holds in the following smaller family $\cB_t$ 
of finite bipartite graphs. More precisely, we mean the following. Let $t\in \N$ be any fixed constant. Then, $\cB_t$  denotes the family of connected finite bipartite graphs which have at least $t$ vertices. Then, on input of an FO-sentence~$\Phi$, it is 
undecidable whether  there exists a graph in $\cB_t$ satisfying~$\Phi$. 
\qed
\end{remark}
\begin{theorem}\label{thm:endmark}    Let $\Phi$ be an MSO-sentence. Then, $\sat_\Phi(\cdot)$
is decidable for marked graphs $F=(V_F,E_F,\mu)$ as inputs where at most one endpoint of each edge is marked.  
\end{theorem}  
\begin{proof} 
  The family $\cG_F$ yields a family of graphs of bounded tree-width. Indeed, the bag size is uniformly bounded by $|V_F|$. The result follows from the papers cited above.  
 \end{proof} 
The next theorem shows in particular that the FO-theory is undecidable if there is an edge where both endpoints are marked, using Trakhtenbrot's Theorem. 
 \begin{theorem}\label{thm:bepm}  
  Let  $F$ be a marked graph
where both endpoints of some marked edge are marked. Then, $\sat_F(\cdot)$ is undecidable. 
  \end{theorem}  
 
 \begin{proof} 
Let $G$ be any finite  connected bipartite graph with at least $n+1$ vertices. 
According to \prref{thm:Fiffi}, the family $\cG_F$ contains a graph $G'$ which is the disjoint union of $G$ and $F$. 
Let $\Psi$ be a first-order sentence which expresses that $F$ appears  and that 
every vertex outside $F$ is part of a connected component which has more vertices than $F$. 
As $F$ is fixed, $\Psi$ is also of constant size, with $\Oh(n)$ many vertex variables. \Ip we can assume that $\Psi$ is of the form $\exists x_1,\dots,x_{n}\Psi'$, this way fixing the vertices of $F$ in $G'$.
Moreover, on input $\Phi\in \FO$ we can construct another FO-sentence 
$\Phi'$ such that $G\models \Phi\iff G'\models \exists x_1,\dots,x_{n}( \Psi'\wedge \Phi')$. Namely, $\Phi'$ speaks about the graph $G'$, disregarding the vertices of $F$.
Let us denote
$\exists x_1,\dots,x_{n}( \Psi'\wedge \Phi')$ by $\phi_F(\Phi)$. Then $\FO'=\set{\phi_F(\Phi)}{\Phi \in \FO}$ is a subset of $\FO$. 
This renders the  satisfiability problem $\set{\phi(\Phi) \in \FO'}{\exists G'\in\cG_F:  G'\models \phi_F(\Phi)}$ to be undecidable by Trakhtenbrot's Theorem as stated in \prref{rem:tra50}. 
Indeed, if it was decidable, then we could decide whether $G\models \Phi$ for any $G\in \cB_{n+1}$.
  \end{proof}

 Some graph properties where the 
problem $\sat(\cG_F,\Phi)$ is trivially decidable are covered by the next theorem, including the problem whether $\cG_F$ contains a  non-planar graph, and various parametrized problems like: ``Is there some $(V_G,E_G)\in \cG_F$ with a clique bigger than $\sqrt{|V_G|}$?''.\footnote{Questions like this one save us from discussing encodings of numbers as a second input parameter. }   
 
 \begin{theorem}\label{thm:btw} 
  Let $F$ be any marked graph and $\Phi$ be a non-trivial graph property  
  such 
  that $G\models \Phi$ \IFF there is  
  a connected component $G'$ of $G$ such that $G'\models\Phi$. 
  Then, the answer to the satisfiability problem 
$\sat(\cG_F,\Phi)$ is always ``Yes'' in the following two cases.    \begin{enumerate} 
   \item There is  some marked self-loop.
   \item The property $\Phi$ is true for some bipartite  
    edge-graph and there is some marked edge where both endpoints are marked. 
  \end{enumerate} 
 \end{theorem} 

 \begin{proof} 
  Since $\Phi$ is non-trivial, there is some finite graph $G$ modeling~$\Phi$. If there  is  some marked self-loop, then there is some $K\in \cG_F$ such that $K$ is a disjoint union of $F$ and $G$ by \prref{thm:Fiffi}. In the second case,  we can choose $G$ to be a bipartite edge-graph. 
  By assumption, $F$ contains a marked edge where both endpoints are marked. Again, there is some $K\in \cG_{F}$
such that $K$ is a disjoint union of $F$ and $G$, so that we can apply \prref{thm:Fiffi}.
 \end{proof} 

\prref{ex:HC} lists a few graph properties which are not covered by \prref{thm:btw}, but nevertheless the satisfiability problem is decidable.   
Recall that graph $G=(V_G,E_G)$ possesses a Hamiltonian cycle if the cycle on $|V_G|$ vertices is a subgraph of~$G$. 
A matching is a collection of edges of a graph such  that no pair of these has any common vertices. A matching is perfect if it contains $|V_G|/2$ many edges. A set $T$ of vertices of a graph $G=(V_G,E_G)$ is called a dominating set\footnote{The notation $T$ refers to the German notion \emph{Träger}.}
if each vertex $u\in V_G$ has a vertex of $T$ in its \emph{closed neighborhood} $N_G[u] = \os{u} \cup \set{v}{(u,v)\in E}$. 
A set $D$ of vertices of a graph $G=(V_G,E_G)$ is called a
a \emph{defensive alliance} in \cite{KriHedHed2004} if it is non-empty and  each vertex $u\in D$ has at least half of its closed neighborhood within $D$.
 
 \begin{example}\label{ex:HC} 
  Let $F=(V_F,E_F,\mu)$ denote a marked graph as input. 
  Then the following problems are decidable.  
  \begin{enumerate} 
   \item  
   Is there some $G\in \cG_F$ with a Hamiltonian cycle? 
   \item Is there some $G\in \cG_F$ with a perfect matching? 
   \item Is there some $G\in \cG_F$ with a dominating set of size at most $\log_2 |V_G|$? 
     \item Is there some $G\in \cG_F$ with a defensive alliance of size at most $\log_2 |V_G|$?
  \end{enumerate} 
We explain these concrete examples one by one. 

\smallskip  
\noindent
\textbf{Hamiltonian cycle.} In order to decide the existence of some $G\in \cG_F$ with a Hamiltonian cycle, we proceed as follows. Without restriction, we may assume that $F$ has no Hamiltonian cycle, because otherwise we are done. If there is any $G=(V_G,E_G)\in \cG_F$ with a Hamiltonian cycle $Z=(V_G,Z_G)$, then starting at any fixed vertex of~$F$, the cycle yields a linear order on the vertices in $G$ and, by restriction, a linear order on~$V_F$.  Since $F$ is without any Hamiltonian cycle, the cycle leaves  
  $F$ at some vertex $u_1\in V_F$ and reenters $F$ at some vertex $v_1\in V_F$. Continuing this way, we obtain a sequence of pairs  
  $(u_1,v_1)\lds (u_k,v_k)$ with $1\leq k \leq |V_F|$ before the cycle is closed. Let us look at the directed path $u_i=w_0,w_1\lds w_\ell=v_i$ on the cycle starting at some $u_i$ and ending in $v_i$ for some pair $(u_i,v_i)$ with $1\leq i \leq k$. Suppose that  
  $w_r=ab^ca$ and $w_s=ab^da$ for $1\leq r<s <\ell$ with $t\leq c<d$ and $[b^c]=[b^d]$. Then we can modify the graph $G$ as follows: we remove all vertices $w_{r+1}\lds w_s$ from $G$, and we introduce   
  an edge $(w_{r}, w_{s+1})$. In this way, we obtain a smaller graph  
  $G'\in \cG_F$ which still has a Hamiltonian cycle. 
  Thus, if $\cG_F$ contains any graph with a Hamiltonian cycle, then $\cG_F$ contains a graph with at most $p|V_F|$ vertices.  
  Hence, it is enough to enumerate all graphs that have at most $p|V_F|$ vertices which have a Hamiltonian cycle and to check if any of them appears in~$\cG_F$.

\smallskip  
\noindent
  \textbf{Perfect matching.} Let $V_F=\os{x_1\lds x_n}$ and suppose that some  $G=(V_G,E_G)\in \cG_F$  has a perfect matching. We have $V_F\sse V_G$. Hence,  all  
  $x_i\in V_F$ are matched by vertices $y_i\in V_G$.  
  The induced subgraph $G[V_F\cup \os{y_1\lds y_n}]$ has a perfect matching with at most $2|V_F|$ vertices. As in the precedent example, we enumerate and check all these graphs.

\smallskip  
\noindent  \textbf{Dominating set.} 
If there is no marked edge in $F=(V_F,E_F,\mu)$, then decide whether a dominating set with the desired property exists in $(V_F,E_F)$. In the second case, there is a marked vertex which is is an endpoint of a marked edge. Then, $\cG_F$ contains a graph $G$ which is the (not disjoint) union of $F$ and an arbitrarily large star.
The intersection of $F$ and the star is just one point. Thus, we find a graph $G$ in $\cG_F$ 
where $G$ has a dominating set of size $D(G)$ such that $D(G)\leq \log_2|V_G|$. Actually, for every $\eps>0$, there is some  $G\in \cG_F$
such that $D(G)/|V_G|<\eps$.
Thus, in the second case, we return ``Yes.''

\smallskip  
\noindent  \textbf{Defensive alliance.} 
If $F=(V_F,E_F,\mu)$ contains no marked vertex at all, then $\cG_F=\{(V_F,E_F)\}$, so we have to check if $(V_F,E_F)$ contains a sufficiently small defensive alliance. Otherwise, we return ``Yes.'' Namely, in this case $\cG_F$ contains a graph $G$ that consists of $(V_F,E_F)$ plus $2^{|V_F|}$ many isolated vertices. Now,
one of these isolated vertices forms a sufficiently small defensive alliance by itself. \qed
 \end{example}

Frequently, we are not only interested in decision problems, but in computational problems. We illustrate this by computing the supremum 
(in $\N_\infty$) of the chromatic numbers over all the graphs in $\cG_F$.
Recall that a graph $G=(V,E)$ is \emph{$k$-colorable} if there is a function $c:V\to \os{1\lds k}$ such that $(u,v)\in E\sm \id_V$ implies $c(u)\neq c(v)$. Indeed, self-loops should not have any influence on the chromatic number, because otherwise a graph with a self-loop could not be colored at all. For a finite graph, its  \emph{chromatic number} $\chi(G)$ is the minimal possible $k\in\N$ such that $G$ is $k$-colorable.

\begin{proposition}\label{prop:chrom}
Let $F=(V_F,E_F,\mu)$ be a marked graph. Then, 
$\sup\set{\chi(G)}{G\in \cG_F}=\infty$ \IFF $F$ contains a marked self-loop. 
If $F$ is without any marked self-loop, then $\chi(V_F,E_F)= \max\set{\chi(G)}{G\in \cG_F}$ is a natural number.\\
Moreover, if $L\sse \G$ is regular, then 
$\sup\set{\chi(G)}{G\in \rho(L)}$ is effectively computable.\end{proposition}
\begin{proof}Since $F\in \cG_F$, we have $\chi(V_F,E_F)\leq \sup\set{\chi(G)}{G\in \cG_F}$.
If $F=(V_F,E_F,\mu)$ has a marked self-loop, then $\cG_F$ contains 
for each $k\in \N$ a graph having a clique of size $k$ as a subgraph. Hence, $\sup\set{\chi(G)}{G\in \cG_F}=\infty$. 
Therefore, for the rest of the proof, we may assume that 
$F=(V_F,E_F,\mu)$ has no marked self-loops.
For every $G\in \cG_F$, there is a graph morphism $\phi:G\to (V_F,E_F)$. If $F=(V_F,E_F,\mu)$ has no self-loops at all, then 
we have $\chi(G)\leq \chi(F)$, because every fiber $\oi \phi(v)$ is without any edge for $v\in V_F$. Thus, a $k$-coloring of $(V_F,E_F)$ induces a $k$-coloring of $G$. If there is a self-loop around a vertex $v$, then the loop is not marked by assumption. This loop is the only edge in $\oi \phi(v)$, but, by definition, this has no influence on the chromatic number. 
\end{proof}

Together with the results above, we have a meta-theorem for graph properties $\Phi$ with a decidable satisfiability problem, covering all cases where we have positive results.

\begin{theorem}\label{thm:meta}
Let $r:\N\to \N$ be  a non-decreasing computable function and 
$\Phi$ be a graph property such that, for each marked graph $F$, 
the following property holds. If some graph $G=(V,E)\in \cG_F$ satisfies $\Phi$, then there is some graph $G=(V,E)\in \cG_F$ such that $G$ satisfies~$\Phi$ and $G$ has at most $r(|V_F|)$ vertices. Then, given as input a \cfree grammar for a language $L\sse \G$ satisfying the \btptorp, the following satisfiability problem
\begin{align*}
\sat(\rho(L),\Phi)= \text{``}\exists G\in \rho(L):\, G\models \Phi\text{?''}
\end{align*}
is decidable. 
\end{theorem}
\begin{proof}
Since $L$ is \cfree satisfying the \btptorp, we find a regular language $R$ such that 
$\rho(R)=\rho(L)$. Splitting $R$ into finitely many cases, we are reduced to show the claim when the input is a single marked graph
$F=(V_F,E_F,\mu)$. Taking $F$ as input, we compute $n=r(|V_F|)$ and we compute the list of all graphs with at most $n$ vertices. Then, we check whether any graph in that list belongs to the family $\cG_F$ and satisfies~$\Phi$, which is possible thanks to \prref{cor:allreg} (and as $\Phi$ is decidable). 
\end{proof}

\section{Conclusion and open problems}\label{sec:conc}
The starting point of our paper was the following idea: Decide a graph property $\Phi$ not for a single instance as in traditional algorithmic graph theory, but generalize this question to a set of graphs specified by a regular language over a binary alphabet. ``Let's talk about a regular family of graphs\footnote{A first song about a remotely similar theme was released in 1991 by Salt ’n’ Pepa.}, reader''.
We chose a natural representation of graphs by words over a binary alphabet $\Sigma$. Our results are rather robust, other ``natural choices'' work as well. Next, pick your favorite graph property~$\Phi$. For example, $\Phi$ says that 
the number of vertices is a prime number. The property does not look very regular, there is no way to express the property, say, in MSO. Still, given a \cfree language $L\sse \Sigma^*$  which satisfies the \btorp and which encodes sets of graphs, we can answer the question if there exists a graph represented by $L$ and which satisfies~$\Phi$.
This is a consequence of \prref{thm:meta} and Bertrand's postulate that for all $n\geq 1$, there is a prime between~$n$ and~$2n$.  

Still, various problems are open. For example, is the satisfiability problem decidable for graph properties which are not covered by \prref{thm:meta}? This could mean that on input of a marked graph $(F,\mu)$, we can say ``YES, there is such a graph in $\cG_F$'' without producing a witness graph in $\cG_F$ for this claim.

Another type of problems relates to model checking. Given a graph 
property $\Phi$, we can define $\cG(\Phi) = \set{G \text{ is a finite graph}}{G\models \Phi}$. Suppose that $\oi\rho (\cG(\Phi))$ is regular. Given a regular language $R\sse \Sigma^*$, can we decide 
whether $\cG(\Phi) \sse \rho(R)$? What about the equality $\cG(\Phi) = \rho(R)$? We can ask the same two questions if $R$ is \cfree.

Another area which we did not touch at all concerns complexity. We can state however an $\NP$ lower bound for $\NP$-hard 
 graph  properties. 
Observe that our encoding of graphs by words is essentially optimal if we  write exponents $i$ which appear in factors $ab^i a$ in binary. We let $\absbin{ab^i a}= 2+\log_2(i)$, and this induces
a binary length $\absbin{w}$ for $w\in \G$ and also a natural binary length $\absbin{F}$ for marked graphs $F=(V_F,E_F,\mu)$. If  $\Phi$ denotes an $\NP$-hard graph property, then the problem $\set{F=(V_F,E_F,\mu)}{\exists G\in \cG_F: \, G\models \Phi}$ (with binary input size for $F$) is $\NP$-hard. 
It is however not clear that it can be solved within $\NP$ assuming that $\Phi$ is in $\NP$.

\subsection*{Acknowledgement}
We thank Dietrich Kuske for pointing out that the formulation of
\cite[Thm.~3]{DiekertFW_DLT21} is not correct as published in the proceedings. The correct statement is now \prref{thm:bepm}.
We also thank the anonymous referees of DLT'21 for various suggestions to improve the presentation.


\begin{thebibliography}{10}

\bibitem{AndersonLRSS09}
T.~Anderson, J.~Loftus, N.~Rampersad, N.~Santean, and J.~Shallit.
\newblock Detecting palindromes, patterns and borders in regular languages.
\newblock {\em Information and Computation}, 207:1096--1118, 2009.

\bibitem{anisimov71}
A.~V. Anisimov.
\newblock Group languages.
\newblock {\em Kibernetika}, 4:18--24, 1971.
\newblock English translation in Cybernetics and Systems Analysis 4 (1973),
  594-601.

\bibitem{BerMah2016}
S.~Bera and K.~Mahalingam.
\newblock Structural properties of word representable graphs.
\newblock {\em Mathematics in Computer Science}, 10:209--222, 2016.

\bibitem{CourcelleHandbookGG97}
B.~Courcelle.
\newblock The expression of graph properties and graph transformations in
  {Monadic} {Second-Order} {Logic}.
\newblock In G.~Rozenberg, editor, {\em Handbook of Graph Grammars and
  Computing by Graph Transformations, Vol.~1: Foundations}, pages 313--400.
  World Scientific, 1997.

\bibitem{CourcelleEngelfriet2012}
B.~Courcelle and J.~Engelfriet.
\newblock {\em Graph Structure and Monadic Second-Order Logic - {A}
  Language-Theoretic Approach}, volume 138 of {\em Encyclopedia of mathematics
  and its applications}.
\newblock Cambridge University Press, 2012.

\bibitem{MeloO_CSR2020}
A.~A. de~Melo and M.~de~Oliveira~Oliveira.
\newblock Second-order finite automata.
\newblock In H.~Fernau, editor, {\em Computer Science - Theory and Applications
  - 15th International Computer Science Symposium in Russia, {CSR} 2020,
  Yekaterinburg, Russia, June 29 - July 3, 2020, Proceedings}, volume 12159 of
  {\em Lecture Notes in Computer Science}, pages 46--63. Springer, 2020.

\bibitem{DiekertFW_DLT21}
V.~Diekert, H.~Fernau, and P.~Wolf.
\newblock Properties of graphs specified by a regular language.
\newblock In N.~Moreira and R.~Reis, editors, {\em Developments in Language
  Theory - 25th International Conference, {DLT} 2021, Porto, Portugal, August
  16-20, 2021, Proceedings}, volume 12811 of {\em Lecture Notes in Computer
  Science}, pages 117--129. Springer, 2021.

\bibitem{diestel2012}
R.~Diestel.
\newblock {\em Graph Theory, 4th Edition}, volume 173 of {\em Graduate texts in
  mathematics}.
\newblock Springer, 2012.

\bibitem{eil74}
S.~Eilenberg.
\newblock {\em Automata, Languages, and Machines}, volume~A.
\newblock Academic Press, New York and London, 1974.

\bibitem{gs66}
S.~Ginsburg and E.~H. Spanier.
\newblock Semigroups, {P}resburger formulas and languages.
\newblock {\em Pacific Journal of Mathematics}, 16:285--296, 1966.

\bibitem{GulerKLW18}
D.~G{\"{u}}ler, A.~Krebs, K.~Lange, and P.~Wolf.
\newblock Deciding regular intersection emptiness of complete problems for
  {PSPACE} and the polynomial hierarchy.
\newblock In S.~T. Klein, C.~Mart{\'{\i}}n{-}Vide, and D.~Shapira, editors,
  {\em Language and Automata Theory and Applications - 12th International
  Conference, {LATA} 2018, Ramat Gan, Israel, April 9-11, 2018, Proceedings},
  volume 10792 of {\em Lecture Notes in Computer Science}, pages 156--168.
  Springer, 2018.

\bibitem{KharlampovichJA1995}
O.~Kharlampovich.
\newblock {The Word Problem for the {Burnside} Varieties}.
\newblock {\em Journal of Algebra}, 173:613--621, 1995.

\bibitem{KitSei2008}
S.~Kitaev and S.~Seif.
\newblock Word problem of the {Perkins} semigroup via directed acyclic graphs.
\newblock {\em Order}, 25:177--194, 2008.

\bibitem{KriHedHed2004}
P.~Kristiansen, S.~M. Hedetniemi, and S.~T. Hedetniemi.
\newblock Alliances in graphs.
\newblock {\em Journal of Combinatorial Mathematics and Combinatorial
  Computing}, 48:157--177, 2004.

\bibitem{KuskeDLT21}
D.~Kuske.
\newblock Second-order finite automata: {Expressive} power and simple proofs
  using automatic structures.
\newblock In N.~Moreira and R.~Reis, editors, {\em Developments in Language
  Theory - 25th International Conference, {DLT} 2021, Porto, Portugal, August
  16-20, 2021, Proceedings}, volume 12811 of {\em Lecture Notes in Computer
  Science}, pages 242--254. Springer, 2021.

\bibitem{LarMof2000}
N.~J. Larsson and A.~Moffat.
\newblock Off-line dictionary-based compression.
\newblock {\em Proceedings of the {IEEE}}, 88:1722--1732, 2000.

\bibitem{LohManMen2013}
M.~Lohrey, S.~Maneth, and R.~Mennicke.
\newblock {XML} tree structure compression using {RePair}.
\newblock {\em Information Systems}, 38:1150--1167, 2013.

\bibitem{mp71}
R.~McNaughton and S.~Papert.
\newblock {\em Counter-Free Automata}.
\newblock The MIT Press, Cambridge, Mass., 1971.

\bibitem{ms83}
D.~E. Muller and P.~E. Schupp.
\newblock Groups, the theory of ends, and context-free languages.
\newblock {\em Journal of Computer and System Sciences}, 26:295--310, 1983.

\bibitem{parikh66}
R.~J. Parikh.
\newblock On context-free languages.
\newblock {\em Journal of the ACM}, 13:570–581, 1966.

\bibitem{RobSey84}
N.~Robertson and P.~D. Seymour.
\newblock Graph minors. {III}. {P}lanar tree-width.
\newblock {\em Journal of Combinatorial Theory}, 36:49--64, 1984.

\bibitem{sch65sf}
M.-P. Sch{\"u}tzenberger.
\newblock On finite monoids having only trivial subgroups.
\newblock {\em Information and Control}, 8:190--194, 1965.

\bibitem{Seese91}
D.~Seese.
\newblock The structure of the models of decidable monadic theories of graphs.
\newblock {\em Annals of Pure and Applied Logic}, 53:169--195, 1991.

\bibitem{Tra50}
B.~A. Trahtenbrot.
\newblock The impossibility of an algorithm for the decision problem for finite
  domains (in {Russian}).
\newblock {\em Doklady Akademii Nauk SSSR}, New Series, 70:569--572, 1950.
\newblock English Translation in American Mathematical Society, Translations
  (1963), Vol.~23, pages 1-5.

\bibitem{VyalyiR15}
M.~N. Vyalyi and A.~A. Rubtsov.
\newblock On regular realizability problems for context-free languages.
\newblock {\em Problems of Information Transmission}, 51:349--360, 2015.

\bibitem{Wolf2019}
P.~Wolf.
\newblock On the decidability of finding a positive {ILP}-instance in a regular
  set of {ILP}-instances.
\newblock In M.~Hospod{\'{a}}r, G.~Jir{\'{a}}skov{\'{a}}, and
  S.~Konstantinidis, editors, {\em Descriptional Complexity of Formal Systems -
  21st {IFIP} {WG} 1.02 International Conference, {DCFS} 2019, Ko{\v{s}}ice,
  Slovakia, July 17-19, 2019, Proceedings}, volume 11612 of {\em Lecture Notes
  in Computer Science}, pages 272--284. Springer, 2019.

\bibitem{Wolf2020}
P.~Wolf.
\newblock From decidability to undecidability by considering regular sets of
  instances.
\newblock In G.~Cordasco, L.~Gargano, and A.~A. Rescigno, editors, {\em
  Proceedings of the 21st Italian Conference on Theoretical Computer Science,
  Ischia, Italy, September 14-16, 2020}, volume 2756 of {\em {CEUR} Workshop
  Proceedings}, pages 33--46. CEUR-WS.org, 2020.

\end{thebibliography}
\newcommand{\Ju}{Ju}\newcommand{\Ph}{Ph}\newcommand{\Th}{Th}\newcommand{\Ch}{Ch}\newcommand{\Yu}{Yu}\newcommand{\Zh}{Zh}\newcommand{\St}{St}\newcommand{\curlybraces}[1]{\{#1\}}

\end{document}